\documentclass{amsart}
\usepackage{amssymb,amsthm,amsmath,xcolor}

\newtheorem{Theorem}{Theorem}[section]

\newtheorem{Lemma}[Theorem]{Lemma}
\newtheorem{Proposition}[Theorem]{Proposition}
\newtheorem{Corollary}[Theorem]{Corollary}
\theoremstyle{definition}
\newtheorem{Definition}[Theorem]{Definition}
\theoremstyle{remark}
\newtheorem{rem}[Theorem]{Remark}
\numberwithin{equation}{section}
\newcommand{\R}{\mathbb R}
\newcommand{\N}{\mathbb N}
\newcommand{\Z}{\mathbb Z}

\newcommand{\C}{\mathcal{C}}

\newcommand{\p}{\mathcal{P}}
\newcommand{\A}{\mathcal{A}}
\newcommand{\rel}{\mathcal{R}}

\newcommand{\F}{\mathcal{F}}

\newcommand{\Hol}{\mathcal{H}}
\newcommand{\D}{\mathcal{D}}
\def\d{\partial}
\def\K{{ \! \rm \ I\!K}}
\def\pathloc{ \operatorname{Paths}_{\operatorname{loc}} }
\title[KP, Mulase factorization, and Fr\"olicher Lie groups]{Well-posedness of the 
Kadomtsev-Petviashvili hierarchy, Mulase factorization, and Fr\"olicher Lie groups}
\author[J.-P. Magnot and E.G. Reyes]{Jean-Pierre Magnot$^1$ and Enrique G. Reyes$^2$}

\address{\small $^1$: LAREMA, Universit\'e d’Angers, 2 Bd Lavoisier, 
49045 Angers cedex 1, France and Lyc\'ee Jeanne d'Arc, Avenue de Grande Bretagne, 63000 
Clermont-Ferrand, France}
\email{\small magnot@math.univ-angers.fr ; jean-pierr.magnot@ac-clermont.fr}
\address{\small $^2$:
Departamento de Matem\'{a}tica y Ciencia de la Computaci\'{o}n,
Universidad de Santiago de Chile, Casilla 307 Correo 2, Santiago,
Chile. }\email{\small enrique.reyes@usach.cl ;
e\_g\_reyes@yahoo.ca}

\begin{document}

    \begin{abstract}
We recall the notions of Fr\"olicher and diffeological spaces and
we build regular Fr\"olicher Lie groups and Lie algebras of formal pseudo-differential 
operators in one independent variable. Combining these constructions with a smooth version 
of Mulase's deep algebraic factorization of infinite dimensional groups based on formal 
pseudo-differential operators, we
present two proofs of the well-posedness of the Cauchy problem for the 
Kadomtsev-Petviashvili (KP) hierarchy in a smooth category. We also
generalize these results to a KP hierarchy modelled on formal pseudo-differential 
operators with coefficients which are series in formal parameters, we describe a rigorous 
derivation of the Hamiltonian interpretation of the KP hierarchy, and we discuss how 
solutions depending on formal parameters can lead to sequences of functions converging to 
a class of solutions of the standard KP-II  equation.
    \end{abstract}

\maketitle

\hfill {\large Dedicated to the memory of Professor Leonid Aleksandrovich Dickey}

\vspace{.6cm}

\textit{Keywords:} KP hierarchy, Mulase factorization, Fr\"olicher Lie groups and Lie
 algebras; Well-posedness; Ambrose-Singer theorem.

\textit{MSC(2010):} 35Q51; 37K10; 37K25; 37K30

\section{Introduction}
In the 1980's M. Mulase published two fundamental papers on the
algebraic structure and formal integrability properties of the
Kadomtsev-Petviashvili (KP) hierarchy, see \cite{M1,M3}. A common
theme in these papers was the use of a powerful theorem on the
factorization of a group of formal pseudo-differential operators
of infinite order which integrated  the algebra of formal
pseudo-differential operators: this factorization ---a delicate
algebraic generalization of the Birkhoff decomposition of loop
groups appearing for example in \cite{GR,PS}--- allowed him to
{\em solve} the Cauchy problem for the KP hierarchy in an
algebraic context.

In this paper we adapt Mulase's results and constructions to a
rigorous setting suitable for Analysis, and we prove
well-posedness of the Cauchy problem for the KP hierarchy in a
smooth category. In fact, we provide two solutions for the Cauchy
problem. The first one is modelled after the theory of $r$-matrices
\cite{RS,STS}, while the second one uses an infinite-dimensional
version of the Ambrose-Singer theorem (see \cite{Ma2013,Ma2015}, and \cite{Stern} for the 
classical finite-dimensional case).
%and it uses explicitly a smooth version of the Mulase factorization. The 
%second one is based on an infinite-dimensional
%version of the Ambrose-Singer theorem (see \cite{Ma2013,Ma2015}, and \cite{Stern} for the 
%classical finite-dimensional case) and a factorization result. 
As our smooth category we choose the setting of regular Fr\"olicher Lie
groups and algebras, see \cite{KM}, which is not too different
from the so-called convenient setting described in the same
reference. This context allows us to construct genuine Lie groups
and Lie algebras structures out of spaces of formal
pseudo-differential operators.

It is well-known that endowing spaces of formal
pseudo-differential operators with rigorous manifold structures
derived from topological structures on infinite dimensional vector
spaces, is a non-trivial issue. For example, we may recall that
there is no natural Lie group attached to the algebra of
differential operators on $S^1$ (see \cite[Sect. 4]{KZ},
\cite[II.4.4]{KW}, or \cite[Chp. 10]{GR}), and that the manifolds
modelled on ``direct limit of ILH algebras'' appearing in
\cite[Section 5]{ARS1} possess underlying topological structures
which need to be treated with extreme care, since there appear,
e.g., locally convex topological vector spaces which are not
complete. We need groups such as the one mentioned above in order
to establish a smooth version of Mulase's factorization, our main
tool for proving well-posedness. Our constructions in a
Fr\"olicher setting are indeed flexible and user-friendly enough
so that they allow us to distinguish clearly which properties
depend on smoothness or on topology. In our approach via
Fr\"olicher Lie groups, properties where topology appears, such as
integration, are clearly delimited. On the other hand, properties
which we obtain via differentiation, which are numerous, are also
explicitly described. Moreover, see Remark \ref{comp}, our chosen
approach is compatible with more standard settings, in the sense
that smoothness in the Fr\"olicher category coincides with
smoothness in the (restricted) category of the
($c^\infty$-)convenient setting \cite{KM} and with Gateaux
smoothness for Fr\'echet manifolds. These facts allow us to carry
out, without using topological arguments which could lead us to
restricting the field of applications, explicit computations
leading, for example, to the announced smooth version of Mulase's
factorization, to solving the Cauchy problem and to a rigorous
hamiltonian formulation for the KP hierarchy. They also allow us
to examine explicit examples: motivated by the theory of
pseudo-differential operators with rough coefficients, see
\cite{Marsch1988}, we propose a deformed KP hierarchy and we
remark briefly that its solutions determine solutions to the
KP-II equation as it appears for instance in \cite{Bourg1993,Sh1986}.

Our paper is structured as follows. Section \ref{S2} is a summary of the
necessary notions in the theory of diffeological and of Fr\"olicher spaces
which we will use in this work\footnote{As explained in the 
previous paragraphs, our favoured setting is the category of Fr\"olicher 
spaces. This category is a full subcategory of the category of diffeological
spaces, and therefore it is natural to begin by considering the latter 
spaces, as we do in Section \ref{S2}.} and it is also a contribution to
the theory of regular Fr\"olicher Lie groups and Lie algebras, as
we now explain. We introduce and study the main properties of
Souriau's diffeological spaces and Fr\"olicher spaces in the first
five subsections of Section 2. Then, we introduce diffeological
groups and the more refined notion of Fr\"olicher Lie groups.
Interestingly, it is not obvious how to define tangent spaces
--let alone Lie algebra-- of a diffeological group. We do it in
detail and we prove (see Proposition \ref{leslie}) that in fact
they are diffeological vector spaces. This proposition is a
precise version of a result announced in \cite{Les}. We then
consider (regular) Fr\"olicher Lie groups and (regular)
Fr\"olicher Lie algebras. In the context of
Lie groups modelled on e.g. Fr\'echet spaces, this notion has led to many 
investigations, see \cite{ARS1,ARS2,DGV,GV,Glo,Ma2018-2,Om,OMYK2}, and also 
\cite{Neeb2007} for an overview of presently open problems. We discuss 
regularity in subsections 2.6 and 2.7, and we answer a question
raised by Kriegl and Michor  in \cite{KM} (after Omori \cite{Om})
on the existence of a non-regular (in a precise sense appearing in
2.7 below) Lie group. We finish Section 2 with a summary of the
theory of principal fiber bundles and the Ambrose-Singer theorem
in the context of diffeological spaces after \cite{Ma2013}. We
need these tools for our second proof of the well-posedness
of the Cauchy problem of the KP hierarchy.

Section \ref{S3} is on the algebra of formal pseudo-differential
operators in one independent variable and its ``integration" to a
regular Fr\"olicher Lie group. We begin with a review of some
aspects of Mulase's work including his factorization theorem, see
\cite{M1,M3} and a more recent review \cite{ER2013}, and then we adapt his 
constructions so as to obtain
a regular Fr\"olicher Lie group $G(\overline{\Psi}(A_t))$ of
infinite order formal pseudo-differential operators, with regular
Fr\"olicher Lie algebra.
This Lie algebra is given by formal pseudo-differential operators
with appropriate coefficients which are reminiscent of Mulase's
choice of ``time-dependent" coefficients, see \cite{M1,M3}. These 
constructions allow us to prove a smooth version of Mulase's factorization
theorem, our main tool in the solution of the
Cauchy problem for the KP hierarchy.

The actual analysis of the Cauchy problem for the KP hierarchy in
a smooth category is carried out in Section \ref{S4}. As announced, we
present two proofs of well-posedness: one uses an
infinite-dimensional version of the Reyman--Semenov-Tian-Shansky
integration theorem, see \cite{RS}, and the other uses
infinite-dimensional principal bundles and the version of the
Ambrose-Singer theorem given in Section 2. Our proofs are much
enriched versions of previous work by the authors, see
\cite{ER2013,Ma2013}: in \cite{ER2013} is shown that the
Reyman--Semenov-Tian-Shansky integration theorem yields formally
solutions to the equations of the KP hierarchy via Mulase's
factorization theorem, but no comments are made in that paper with
respect to the existence of smooth solutions. On the other hand,
the second proof builds upon tools developed for the study of a
$q$-deformed version of the KP hierarchy in \cite{Ma2013,Ma2015},
and completes these two works.

The KP hierarchy with formal series coefficients, with emphasis on
solutions which are functions with low regularity, is considered
on its own right in Section \ref{S5}. We argue that the deformed initial
data is not only an appropriate tool for the analysis of standard
KP, but also that it is a natural way to study some dynamical
systems with rough initial data. In fact, we expect that our work
can be used to investigate actual existence of conservation laws
or blow-up phenomena. We also consider in this section the
hamiltonian structure of the KP and deformed KP hierarchies,
extending to a smooth setting the approach of \cite{ER2013}. This
structure has been studied before for non-deformed KP hierarchy
and functions in $C^\infty(S^1,\R)$ (see {\em e.g.}, \cite{W1,W2,D,KW}), but
we propose a general non-formal version valid in our smooth
category of Fr\"olicher spaces. As a final remark, in Subsection 5.3 we point out the need
to compare the solutions to the KP-II equation which are deduced from our solutions to the
{\em deformed} KP hierarchy, with the ones considered in
\cite{Bourg1993,Sh1986}. We believe this comparison is an important open problem in the
theory of non-formal analysis of nonlinear integrable systems.

\smallskip

%\smallskip

We close this introduction on a personal note. The second named author was fortunate to be 
Prof. Dickey's colleague at the University of Oklahoma during two years. Leonid 
Aleksandrovich was a kind and generous human being, in addition to being a brilliant 
mathematician. It is an honour to dedicate this work to his memory. 

\section{Preliminaries on diffeological Lie groups} \label{S2}
\subsection{Souriau's diffeological spaces and Fr\"olicher spaces}
\label{1.1}

\begin{Definition} \cite{Sou}, see e.g. \cite{Igdiff}. Let $X$ be a set.

\noindent $\bullet$ A \textbf{p-parametrization} of dimension $p$ 
(or $p$-plot) on $X$ is a map from an open subset $O$ of $\R^{p}$ to $X$.

\noindent $\bullet$ A \textbf{diffeology} on $X$ is a set $\p$
of parametrizations on $X$, called plots of the diffeology, such that, for 
all $p\in\N$,

    - any constant map $\R^{p}\rightarrow X$ is in $\p$;

    - Let $I$ be an arbitrary set of indexes; let 
    $\{f_{i}:O_{i}\rightarrow X\}_{i\in I}$ be a family of compatible maps 
    that extend to a map $f:\bigcup_{i\in I}O_{i}\rightarrow X$.
    If $\{f_{i}:O_{i}\rightarrow X\}_{i\in I}\subset\p$, then $f\in\p$.

    - Let $f\in\p$, defined on $O\subset\R^{p}$. Let $q\in\N$,
    $O'$ an open subset of $\R^{q}$ and $g$ a smooth map (in the usual
    sense) from $O'$ to $O$. Then, $f\circ g\in\p$.

\noindent $\bullet$ If $\p$ is a diffeology on $X$, then $(X,\p)$ is
    called a \textbf{diffeological space}.

    \noindent Let $(X,\p)$ and $(X',\p')$ be two diffeological spaces;
    a map $f:X\rightarrow X'$ is \textbf{differentiable} (=smooth) if
    and only if $f\circ\p\subset\p'$. \end{Definition}

\begin{rem}
Any diffeological space $(X,\p)$ can be
endowed with the  weakest topology such that all the maps that
belong to $\p$ are continuous. We do not dwell deeper on this fact
in this work because it is not closely related to the main themes
of this paper.
\end{rem}

%\vskip 9pt

We now introduce Fr\"olicher spaces, see \cite{FK}, using terminology from \cite{KM}.

\begin{Definition} $\bullet$ A \textbf{Fr\"olicher} space is a triple
    $(X,\F,\C)$ such that

    - $\C$ is a set of paths $\R\rightarrow X$,

    - $\F$ is the set of functions from $X$ to $\R$, such that a function
    $f:X\rightarrow\R$ is in $\F$ if and only if for any
    $c\in\C$, $f\circ c\in C^{\infty}(\R,\R)$;

    - A path $c:\R\rightarrow X$ is in $\C$ (i.e. is a \textbf{contour})
    if and only if for any $f\in\F$, $f\circ c\in C^{\infty}(\R,\R)$.

    \vskip 5pt $\bullet$ Let $(X,\F,\C)$ and $(X',\F',\C')$ be two
    Fr\"olicher spaces; a map $f:X\rightarrow X'$ is \textbf{differentiable}
    (=smooth) if and only if $\F'\circ f\circ\C\subset C^{\infty}(\R,\R)$.
\end{Definition}

Any family of maps $\F_{g}$ from $X$ to $\R$ generates a Fr\"olicher
structure $(X,\F,\C)$ by setting, after \cite{KM}:

- $\C=\{c:\R\rightarrow X\hbox{ such that }\F_{g}\circ c\subset C^{\infty}(\R,\R)\}$

- $\F=\{f:X\rightarrow\R\hbox{ such that }f\circ\C\subset C^{\infty}(\R,\R)\}.$

In this case we call $\F_g$ a \textbf{generating set of functions}
for the Fr\"olicher structure $(X,\F,\C)$. One easily see that
$\F_{g}\subset\F$. This notion is useful for this paper since
it allows us to describe a Fr\"olicher structure in a simple way, see for 
instance Proposition \ref{froproj} below. A Fr\"olicher space $(X,\F,\C)$
carries a natural topology, which is the pull-back topology of
$\R$ via $\F$. We note that in the case of a finite dimensional
differentiable manifold $X$ we can take $\F$ the set of all smooth
maps from $X$ to $\R$, and $\C$ the set of all smooth paths from
$\R$ to $X.$ In this case the underlying topology of the
Fr\"olicher structure is the same as the manifold topology
\cite{KM}. In the infinite dimensional case, there is to our
knowledge no complete study of the relation between the
Fr\"olicher topology and the manifold topology; our intuition is
that these two topologies can differ.

We also remark that if $(X,\F, \C)$ is a Fr\"olicher space, we can
define a natural diffeology on $X$ by using the following family
of maps $f$ defined on open domains $D(f)$ of Euclidean spaces
(see \cite{Ma}):
\begin{equation}  \label{alfa1}
\p_\infty(\F)=
\coprod_{p\in\N}\{\, f: D(f) \rightarrow X; \, \F \circ f \in C^\infty(D(f),\R) \quad \hbox{(in
    the usual sense)}\}\; .
\end{equation}
If $X$ is a differentiable manifold, this diffeology has been
called the {\bf n\'ebuleuse diffeology} by  P. Iglesias-Zemmour, see
\cite{Igdiff}. We can easily show the following:

\begin{Proposition} \label{fd} \cite{Ma}
    Let$(X,\F,\C)$
    and $(X',\F',\C')$ be two Fr\"olicher spaces. A map $f:X\rightarrow X'$
    is smooth in the Fr\"olicher sense if and only if it is smooth for
    the underlying diffeologies $\p_\infty(\F)$ and $\p_\infty(\F').$
\end{Proposition}

Thus, we can also state:
\vskip 12pt

\begin{tabular}{ccccc}
Smooth manifold  & $\Rightarrow$  & Fr\"olicher space  & $\Rightarrow$  & 
Diffeological space
\end{tabular}

\vskip 12pt A deeper analysis of these implications has been given
in \cite{Wa}. The next remark is inspired on this work and on
\cite{Ma}; it is based on \cite[p.26, Boman's theorem]{KM}.

\begin{rem}
    The set of contours $\C$ of a Fr\"olicher space
    $(X,\F,\C)$ \textbf{does not} give us a diffeology, because a diffeology
    needs to be stable under restriction of domains. In the case of paths in
    $\C$ the domain is always $\R.$ However, $\C$ defines a ``minimal 
    diffeology'' $\p_1(\F)$ whose plots are smooth parametrizations which 
    are locally of the type $c \circ g,$ where $g \in \p_\infty(\R)$  and 
    $c \in \C.$ Within this setting,
    we can replace $\p_\infty$ by $\p_1$ in Proposition \ref{fd}.
\end{rem}

We also remark that given an algebraic structure, we can define a
corresponding compatible diffeological structure. For example, a
$\R-$vector space equipped with a diffeology is called a
diffeological vector space if addition and scalar multiplication
are smooth (with respect to the canonical diffeology on $\R$). An
analogous definition holds for Fr\"olicher vector spaces. We 
will also consider diffeological groups, see Section \ref{2.6.}

\begin{rem} \label{comp}
Fr\"olicher, $c^\infty$ (the ``smooth convenient setting" of 
\cite{KM}) and G\^ateaux smoothness are the same notion
if we restrict to a Fr\'echet context, see \cite[Theorem 4.11]{KM}.
Indeed, for a smooth map $f : (F, \p_1(F)) \rightarrow \R$ defined
on a Fr\'echet space with its 1-dimensional diffeology, we have
that $\forall (x,h) \in F^2,$ the map $t \mapsto f(x + th)$ is
smooth as a classical map in $\C^\infty(\R,\R).$ And hence, it is
G\^ateaux smooth. The converse is obvious.
\end{rem}

\begin{Definition}[Smooth Homotopy]\label{d:homotopic}
    Let $X$ and $Y$ be diffeological spaces and let $f_i\colon X\to Y$, $i 
    \in \{0,1\}$ be smooth maps. The maps $f_0$ and $f_1$ are 
    \textbf{(smoothly) homotopic} if there is a smooth map
    $H \colon X \times \R \rightarrow Y$ such that $H(\cdot,0)= f_0$ and 
    $H(\cdot,1) = f_1$.  We call the map $H$ a \textbf{smooth homotopy}.
\end{Definition}

This definition enables us to define straightforwardly smooth
homotopy equivalences, fundamental groups and related objects, see
\cite{Igdiff} for details.

\subsection{Tangent space}
Let $X$ be a diffeological space. There exist two main ways to define the
tangent space at a point $x \in X$.  The \textbf{internal tangent space} at 
$x \in X$ described in \cite{CW}, and the \textbf{external tangent space} 
$^eTX$, defined simply as the set of derivations on $C^\infty(X,\R)$,
see \cite{Igdiff,KM}. It is known that these two constructions coincide in 
the case of finite dimensional manifolds and in other important cases, see 
\cite[Section 28]{KM} and \cite{CW}. For us, it is actually enough
to define \textbf{tangent cones} after \cite{DN2007-1}.
	
	\smallskip
	
	For each $x\in X,$ we consider the set of paths 
	$$
	C_{x}=\{c \in C^\infty(\R,X)| c(0) = x\}
	$$ 
	and take the equivalence relation $\mathcal{R}$ given by 
	$$
	c\mathcal{R}c' \Leftrightarrow 
	\forall f \in C^\infty(X,\R), \quad \partial_t(f \circ c)|_{t = 0} = 
	\partial_t(f \circ c')|_{t = 0}\; .
	$$
	The tangent cone at $x$ is the quotient 
	$$
	^iT_xX = C_x / \mathcal{R}\; .
	$$ 
Equivalence classes of $\mathcal{R}$ are  denoted by 
	$V = \partial_tc(t)|_{t=0} = \partial_t c(0)\in {}^iT_xX$. We also use the notation  
	$$
	Df(V) = \partial_t(f \circ c)|_{t = 0}\; .
	$$
	
	It is shown in \cite[Section 2]{DN2007-1} that there exist examples of diffeological spaces for which 
	the tangent cone at a point $x$ is not a vector space, hence the need for more sophisticated definitions
	as in \cite{CW,Igdiff,KM}. However, in subsection 2.6 we show that this difficulty is absent in the case
	of diffeological groups. 

\subsection{Differential forms}

\begin{Definition} \cite{Sou}
    Let $(X,\p)$ be a diffeological space and let $V$ be a vector space equipped with a 
    differentiable structure. A $V-$valued $n$-differential form $\alpha$ on $X$ 
    (noted $\alpha \in \Omega^n(X,V))$ is a map
    $$ \alpha : \{p:O_p\rightarrow X\} \in \p \mapsto \alpha_p \in \Omega^n(O_p;V)$$
    such that

    $\bullet$ Let $x\in X.$ $\forall (p,p')\in \p^2$ such that $x\in Im(p)\cap Im(p')$,
    the forms $\alpha_p$ and $\alpha_{p'}$ are of the same order $n.$

    $\bullet$ Let $y\in O_p$ and $y'\in O_{p'}$, and assume that $(X_1,...,X_n)$ are $n$
     germs of paths in
    $Im(p)\cap Im(p').$ If there exists two systems of $n-$vectors 
    $(Y_1,...,Y_n)\in (T_yO_p)^n$ and $(Y'_1,...,Y'_n)\in (T_{y'}O_{p'})^n$ such that
    $p_*(Y_1,...,Y_n)=p'_*(Y'_1,...,Y'_n)=(X_1,...,X_n),$ then
    $$ \alpha_p(Y_1,...,Y_n) = \alpha_{p'}(Y'_1,...,Y'_n)\, .$$

    We note by $$\Omega(X;V)=\oplus_{n\in \mathbb{N}} \Omega^n(X,V)$$ the set of 
    $V-$valued differential forms.
\end{Definition}

We make two remarks for the reader:

\smallskip

$\bullet$ If there do not exist $n$ linearly independent vectors
$(Y_1,...,Y_n)$ as in the last point of the definition, then
$\alpha_p = 0$ at $y.$

$\bullet$ Let $(\alpha, p, p') \in \Omega(X,V)\times \p^2.$
If there exists $g \in C^\infty(D(p); D(p'))$ (in the usual sense)
such that $p' \circ g = p,$ then $\alpha_p = g^*\alpha_{p'}.$

\vskip 12pt
\begin{Proposition}
    The set $\p(\Omega^n(X,V))$ of all maps $q:x \mapsto 
    \alpha(x)$ from an open subset $O_q$ of $V$ 
 to $\Omega^n(X,V)$ such that for each $p \in \p,$ 
 $$\{ x \mapsto \alpha_p(x) \} \in C^\infty(O_q,\Omega^n(O_p,V))\; ,$$
    is a diffeology on $\Omega^n(X,V).$
\end{Proposition}

Working on plots of the diffeology, we can define the product and the differential of 
differential forms, which have the same properties as the product and the differential of
differential forms, see for instance
\cite{Igdiff}.

\begin{Definition}
    Let $(X,\p)$ be a diffeological space.

    \noindent
    $\bullet$ $(X,\p)$ is \textbf{finite-dimensional} if and only if 
    $$\exists n_0\in\mathbb{N},\quad \forall n\in \mathbb{N}, \quad 
    n\geq n_0 \Rightarrow dim(\Omega^n(X,\mathbb{R}))=0\; .$$
    Then, we set $$dim(X,\p)=max\{n\in \mathbb{N}| dim(\Omega^n(X,\mathbb{R}))>0\}.$$
    \noindent
    $\bullet$ If not, $(X,\p)$ is called \textbf{infinite dimensional}.
\end{Definition}

Let us make a few remarks on this definition. If $X$ is a manifold
with $dim(X)=n,$ the n\'ebuleuse diffeology is such that
$$dim(X,\p_\infty)=n\; .$$ 
On the other hand, if $(X,\F,\C)$ is the natural
Fr\"olicher structure on the finite-dimensional manifold $X,$ and we take 
$\p_1$ defined as before, then it is 
an easy exercise to show that $$dim(X,\p_1)=1\; .$$ Therefore, the
dimension depends on the diffeology considered. Now, we remark
that $\F$ is not only the set of smooth maps $(X,\p_1)\rightarrow
(\mathbb{R},\p_1(\R)),$ but also the set of smooth maps
$(X,\p_1)\rightarrow (\mathbb{R},\p_\infty(\R)),$ and also the set
of smooth maps $(X,\p_\infty)\rightarrow
(\mathbb{R},\p_\infty(\R)).$ These observations follow from
adapting arguments of \cite{Ma} based on Boman's theorem as it
appears in \cite{KM}, see e.g. \cite{Wa}.

We are led to the following definition, since $\p_\infty(\F)$ is
clearly the diffeology with the largest dimension associated to
$(X,\F,\C)$:

\begin{Definition}
    The \textbf{dimension} of a Fr\"olicher space $(X,\F,\C)$ is
    the dimension of the diffeological space $(X,\p_\infty(\F)).$
\end{Definition}

\subsection{Push-forward, quotient and subsets}

We give here only the results that will be used hereafter.

\begin{Proposition} \cite{Sou,Igdiff} Let $(X,\p)$ be a diffeological space,
    and let $X'$ be a set. Let $f:X\rightarrow X'$ be a map.
    We define the \textbf{push-forward diffeology} as the coarsest (i.e. the 
    smallest for inclusion) among the diffologies
    on $X'$, which contains $f \circ \p.$ \end{Proposition}

We have now the tools needed to describe the diffeology on a quotient:

\begin{Proposition} \label{quotient} Let $(X,\p)$ be a diffeological
    space and let $\rel$ be an equivalence relation on $X$. Then, there is
    a natural diffeology on $X/\rel$, noted by $\p/\rel$, defined as
    the push-forward diffeology on $X/\rel$ induced by the quotient 
    projection $X\rightarrow X/\rel$. 
\end{Proposition}

Given a subset $X_{0}\subset X$, where $X$ is a Fr\"olicher space
or a diffeological space, we can define a subset structure on
$X_{0}$, induced by $X$:

$\bullet$ If $X$ is equipped with a diffeology $\p$, we can define
a diffeology $\p_{0}$ on $X_{0},$ called the \textbf{subset
diffeology} \cite{Sou,Igdiff} by setting 
\[ 
\p_{0}=\lbrace p\in\p \hbox{ such that the image of }p\hbox{ is a subset of
}X_{0}\rbrace\; .
\]

$\bullet$ If $(X,\F,\C)$ is a Fr\"olicher space, we take as a generating
set of maps $\F_{g}$ on $X_{0}$ the restrictions of the maps $f\in\F$.
In that case, the contours (resp. the induced diffeology) on $X_{0}$
are the contours (resp. the plots) on $X$ whose images are 
subsets of $X_{0}$. 

\subsection{Cartesian products and projective limits}

We review the differentiable structure of products and projective limits since they are 
crucial for our discussion of Fr\"olicher Lie groups of formal pseudo-differential 
operators. We begin with the following important proposition (see \cite{Sou}):

\begin{Proposition} \label{prod1} Let $(X,\p)$ and $(X',\p')$
    be two diffeological spaces. There exists a diffeology $\p\times\p'$ on
    $X\times X'$  made of plots $g:O\rightarrow X\times X'$
    that decompose as $g=f\times f'$, where $f:O\rightarrow X\in\p$
    and $f':O\rightarrow X'\in\p'$. We call it the \textbf{product diffeology}.  
\end{Proposition}

We apply this result to the case of Fr\"olicher spaces and we derive very easily, 
(compare with e.g. \cite{KM}) the following:

\begin{Proposition} \label{prod2} 
Let $(X,\F,\C)$ and $(X',\F',\C')$ be two Fr\"olicher spaces equipped with 
their natural diffeologies $\p$ and $\p'$ . There is a natural structure of 
Fr\"olicher space on $X\times X'$ whose contours $\C\times\C'$ are the 1-plots of 
$\p\times\p'$. 
\end{Proposition}

We can even state the results above for infinite products: 
we simply take cartesian products of the plots or of the contours.
Let us now consider projective limits of Fr\"olicher and diffeological spaces.

\begin{Proposition} \label{froproj} 
Let $\Lambda$ be an infinite set of indexes.

   $\bullet$ Let $\lbrace(X_{\alpha},\p_{\alpha})\rbrace_{\alpha\in\Lambda}$
    be a family of diffeological spaces indexed by $\Lambda$ and totally
    ordered by inclusion. We write $\lbrace i_{\beta,\alpha}:        
X_\alpha \rightarrow X_\beta\rbrace_{(\alpha,\beta)\in\Lambda^{2}}$ for 
    the corresponding family of inclusion maps, and we assume that they are smooth. 
    If $X=\bigcap_{\alpha\in\Lambda}X_{\alpha},$ then
    $X$ carries the \textbf{projective diffeology} $\p$, which is the
    pull-back of the diffeologies $\p_{\alpha}$ of each $X_{\alpha}$
    via the family of inclusion maps $\lbrace f_{\alpha}:X \rightarrow X_
    \alpha\rbrace_{\alpha\in\Lambda}.$ The diffeology
    $\p$ is made of all plots $g:O\rightarrow X$ such that for 
    each $\alpha\in\Lambda,$
    \[
    f_{\alpha}\circ g\in\p_{\alpha}\; .\]
    This is the largest diffeology for which the maps $f_{\alpha}$ are
    smooth.

    $\bullet$ Let 
    $\lbrace(X_{\alpha},\F_{\alpha},\C_{\alpha})\rbrace_{\alpha\in\Lambda}$
    be a family of Fr\"olicher spaces indexed by $\Lambda$ and totally ordered by inclusion.  
    There is a natural structure of Fr\"olicher space on $X=\bigcap_{\alpha\in\Lambda}X_{\alpha}$.
    % for which the corresponding contours
    % \[  \C=\bigcap_{\alpha\in\Lambda}\C_{\alpha}   \]
    % are some 1-plots of $\p=\bigcap_{\alpha\in\Lambda}\p_{\alpha}.$ 
    A generating set of functions for this Fr\"olicher structure is the set
    of maps of the type:
    $$
    \bigcup_{\alpha \in \Lambda} \F_{\alpha} \circ f_{\alpha}\; .$$
\end{Proposition}

\subsection{Regular Fr\"olicher Lie groups} \label{2.6.}

This subsection is mostly inspired in \cite{Ma2013}, see also  \cite{Ma2015}.

\begin{Definition}
    Let $G$ be a group equipped with a diffeology $\p.$ We call $G$ a 
    \textbf{diffeological group} if multiplication and inversion are 
    smooth maps.
\end{Definition}

\noindent An analogous definition holds for Fr\"olicher groups.

\smallskip

After Iglesias-Zemmour, see \cite{Igdiff}, we note that an arbitrary 
diffeological group does not necessarily possess a Lie algebra. Now we give 
conditions for the existence of a tangent space at the identity element $e$ 
(of course, $e$ is precisely the identity mapping if    
the group $G$ is a group of transformations).

\begin{Definition}
The diffeological group $G$ is a \textbf{diffeological Lie group} if and only 
if $^iT_eG$ is a diffeological vector space and the derivative of the Adjoint action of $G$ 
on $^iT_eG$ defines a Lie bracket. In this case, we call $^iT_eG$ the Lie 
algebra of $G$, and we denote it generically by $\mathfrak{g}.$
\end{Definition}

%Let us summarize some facts concerning the algebraic, diffeological and 
%Fr\"olicher structures of $\mathfrak{g}.$ 

\begin{rem}
    Let $G$ be a diffeological Lie group with Lie algebra
$$\mathfrak{g} = \{ \partial_t c(0) : c \in C_e \hbox{ and } c(0)=e \}\; . $$
We note that this definition coincides with the classical definition of the Lie 
algebra of a finite dimensional Lie group via germs of paths at $e.$ 
    We have: 
\begin{itemize}
\item Let $(X,Y) \in \mathfrak{g}^2$. Then, $X+Y = \partial_t(c.d)(0)$  where 
    $c,d \in C_e ^2,$ $c(0) = d(0) =e ,$ $X = \partial_t c(0)$ and 
    $Y = \partial_t d(0).$
\item Let $(X,g) \in \mathfrak{g}\times G$. Then, 
$Ad_g(X) = \partial_t(g c g^{-1})(0)$  
      where $c \in C_e ,$ $c(0) =e ,$ and $X = \partial_t c(0).$
\item Let $(X,Y) \in \mathfrak{g}^2$. Then, $[X,Y] = \partial_t( Ad_{c(t)}Y)$   
where $c \in C_e ,$ $c(0) =e ,$ and $X = \partial_t c(0).$
\end{itemize}
All these operations are smooth (and thus well-defined) in the context of 
Fr\"olicher Lie groups and Fr\"olicher Lie algebras as well.
\end{rem}

We now present a complete proof of \cite[Proposition 1.6]{Les} which states that if we start 
with a diffeological group $G$, we only need to check that $^iT_eG$ admits a Lie bracket in 
order for $G$ to be a diffeological Lie group with Lie algebra $^iT_eG$.

\begin{Proposition} \label{leslie}
    Let $G$ be a diffeological group. Then the tangent cone at the identity 
    element, $^iT_eG$, is a diffeological vector space.
\end{Proposition}
\begin{proof}
    Let $(X,Y)\in {}^iT_eG,$ $X = \d_t x(t)|_{t=0}$ and $Y = \d_t y(t)|_{t=0}
    $ From \cite{DN2007-1} and \cite{Les}, see e.g. \cite{Igdiff}, we know that 
    $^iT_eG$ is a cone, and we only have to complete the proof of \cite{Les} 
    to show that the smooth operation 
    $$X+Y = \partial_t (x(t).y(t))|_{t=0}$$ transforms $^iT_eG$ into a vector 
    space (smoothness of the operation $+$ follows from \cite[Lemma 1.2]{Les}). 
    
    First, we remark that 
$$
\forall \lambda \in \R, \quad \lambda X + \lambda Y = 
\partial_t(x(\lambda t).y(\lambda t))|_{t=0} = \lambda \partial_t(x( t).y(t))|_{t=0}
= \lambda(X+Y)\; .
$$
    Thus, we only have to prove that $X+Y = Y + X,$ that is 
    $$\partial_t (x(t).y(t))|_{t=0} = \partial_t (y(t).x(t))|_{t=0}\; .$$
    For this, we assume that $G$ is equipped with its nebuleuse diffeology 
    $\p_\infty$. We recall that this diffeology consists of all smooth maps
    $p : D(p) \subset \R^d \rightarrow G$ where $D(p)$ is equipped with the 
    diffeology $\p_1(D(p))$ generated by
    $$\left\{\gamma \in C^\infty(]a,b[,D(p)) \hbox{ (in the usual sense) } 
      |\, (a,b)\in \R^2 \wedge a<b \right\}\cup \{\{0\} \rightarrow G\}.$$
For the diffeology $\p_\infty,$ the group multiplication and inversion are also 
smooth, since they are based on smooth paths. This fact comes from a 
difficult theorem, \cite[Boman's theorem]{KM}, which implies that 
multiplication and inversion are smooth with respect to $\p_\infty$ if an only 
if they are smooth with respect to $\p_1(G),$ using the obvious notations. This 
change of diffeology is necessary at this stage in order to avoid longer 
arguments in this already very long proof; we will relax this assumption 
momentarily.
\vskip 12pt
\noindent
\underline{\it First step: $\d_tx^{-1}(t)|_{t=0} = (-1).\d_tx(t)|_{t=0}\, .$}

    If $y(t) = (x(t))^{-1},$ then $x(t).y(t)= e,$ $X + Y=0,$ and we get in 
    the same way that $Y + X = 0.$ Let $(-1).X$ be denoted by $-X$. For any 
    $f \in C^\infty(G,\R)$ we have 
 {  $$Df((-X)+ X+Y)=\partial_tf(x(-t).x(t).y(t))|_{t=0}\; .$$}
    Let $a,b,c$ be three smooth paths such that $a(0) = b(0) = c(0) = e.$
    Then the map $$\Phi_3: (t_1,t_2,t_3)\mapsto f(a(t_1).b(t_2).c(t_3))\; ,$$
    defined on an open neighbourhood of $(0,0,0)\subset \R^3,$ is a 
    (classical) smooth map, and 
    $$D_{(0,0,0)}\Phi_3(v_1,v_2,v_3) = \d_tf\circ a(t)|_{t=0}.v_1 + 
    \d_tf\circ b(t)|_{t=0}.v_2 +\d_t f \circ c(t)|_{t=0}.v_3\; .$$
    Let us apply this to $a(t)=x(-t),$ $b(t) = x(t)$ and $c(t) = y(t) = (x(t))^{-1},$ 
    with $v_1=v_2=v_3=1.$ We obtain
    \begin{eqnarray*}
    Df(-X)& = &Df(-X) + Df(X+Y)\\
    & = & Df(-X) + Df(\partial_t (x(t).y(t))|_{t=0}) \\
    & = & Df(-X) + \partial_tf(x(t).y(t))|_{t=0}  \\
    & = &\d_tf\circ x(-t)|_{t=0} + \d_tf\circ x(t)|_{t=0} +\d_t f \circ y(t)|_{t=0}\\
    & = &Df(-X+X) + Df(Y)\\& =& Df(Y)\; .
    \end{eqnarray*}
    Thus $Y =-X$ when $G$ is equipped with $\p_\infty.$
    \vskip 12pt
    \noindent
    \underline{\it second step: $X+Y = Y + X.$}

    Here, $y$ is no longer equal to $x^{-1}.$
    Let $f\in C^\infty(G,\R)$ and let us define $$
    \Phi_4(t_1,t_2,t_3,t_4)=f(a(t_1)b(t_2)c(t_3)d(t_4))$$
    on an open neighbourhood of $(0,0,0,0)\subset \R^4,$ where $a,b,c,d$ are 
    smooth paths such that $a(0)=b(0)=c(0)=d(0)=e.$
    Then, mimicking the arguments of the first step, $\Phi_4$ is smooth, and
     \begin{eqnarray*}
     D_{(0,0,0,0)}\Phi_4(v_1,v_2,v_3,v_4)& = & \d_tf\circ a(t)|_{t=0}.v_1 + 
     \d_tf\circ b(t)|_{t=0}.v_2 \\ && +\d_t f \circ c(t)|_{t=0}.v_3 + \d_t f 
     \circ d(t)|_{t=0}.v_4.\end{eqnarray*}
    Setting $v_1 = v_2 = v_3 = v_4 = 1,$ with $a=x, b=y, c=x^{-1}, d=y^{-1},$ 
    we get
     \begin{eqnarray*}
   { \d_tf(x(t).y(t).x^{-1}(t).y^{-1}(t)) }& = & Df(X)+ Df(Y) + Df(-X) + Df(-Y) \\
     & = & Df(X+Y) + \left(-Df(X+Y)\right) \\
     & = & 0
     \end{eqnarray*}
     on the one hand. And on the other hand,
     \begin{eqnarray*}
        \d_tf(x(t).y(t).x^{-1}(t).y^{-1}(t)) & = & Df(X)+ Df(Y) + Df(-X) + Df(-Y) \\
        & = & \d_tf(x(t)y(t)) + \d_tf\left(x^{-1}(t)y^{-1}(t)\right) \\
        & = & \d_tf(x(t)y(t)) + \d_tf\left((y(t)x(t))^{-1}\right) \\
        & = & Df\left(\,(X+Y)+ ( - (Y+X))\,\right)
     \end{eqnarray*}
     Thus in $^iTeG,$ $$X+Y = Y+X,$$
     and hence $^iT_eG$ is a vector space if $G$ is equipped with the 
     $\p_\infty$ diffeology. 

     \vskip 12pt
     \noindent
     \underline{\it Third step: relaxing the n\'ebuleuse diffeology}
     Let $\p_1(G)$ be the diffeology generated by $0-$ and $1-$dimensional 
     plots of $\p$, and let $\p_\infty(G)$ defined as before.
     The first two steps of the foregoing proof are based on the fact that 
     the maps
     $$(t_1,t_2,t_3) \mapsto a(t_1)b(t_2)c(t_3)$$
     and
     $$ (t_1,t_2,t_3,t_4) \mapsto a(t_1)b(t_2)c(t_3)d(t_4)$$
     are plots of the nebuleuse diffeology. But they are also smooth maps for 
     the $\p_1(G)$ diffeology, and hence if 
     $f \in C^\infty((G,\p_1(G)), \R),$ the maps $\Phi_3 $ and $\Phi_4$ 
     remain smooth, and hence the first and the second step remain valid.
     Let us now consider $\p$ a diffeology on $G$ such that $(G,\p)$
     is a diffeological group.
     Then $$\p_1(G) \subset \p \subset \p_\infty(G)$$
     and $$ C^\infty((G,\p_1(G)),\R) \supset   C^\infty((G,\p),\R) \supset  
     C^\infty((G,\p_\infty(G)),\R) .$$

     We have, $\forall f \in C^\infty((G,\p_1(G)),\R),$ and for all smooth 
     paths $ x(t)$ and $y(t)$ in $\p$ and hence in $\p_1(G)$ as above, 
     $$\d_tx^{-1}(t)|_{t=0} = (-1).\d_tx(t)|_{t=0} $$ and  
     $$Df(X+Y) = Df(Y+X).$$
     This is then true for each $f\in C^\infty((G,\p),\R),$ and this fact
      completes the proof.
\end{proof}

Now we go back to the problem of equipping $^iT_eG$ with a Lie algebra structure. 
We remark that, actually, the condition on the existence of a Lie
bracket obtained from  differentiation of the adjoint action of
the group $G$ on ${}^iT_eG$ cannot be relaxed: it has to be
assumed in order to give a structure of Lie algebra to ${}^iT_eG.$
However, if both $G$ and $^iT_eG$ are smoothly embedded into a
diffeological algebra $A,$ and if group multiplication and
$^{i}T_eG-$addition are the pull-back of the operations on $A,$
then the Lie bracket exists and is defined by the standard
relation 
\begin{equation} \label{bra}
[u,v]=uv-vu
\end{equation}
 as long as $^iT_eG$ is stable under (\ref{bra}). We can
 ask the following natural question, somewhat reminiscent
of Hilbert's classical problem on the relation between topological
groups and Lie groups:

\vskip 12pt \centerline{\it Does there exist a diffeological (or
Fr\"olicher) group}

\centerline{\it which is not a diffeological (or Fr\"olicher) Lie group ?}
\vskip 12pt

Let us now concentrate on diffeological and Fr\"olicher Lie groups. We write
 $\mathfrak{g}={}^iT_eG.$
The basic properties of adjoint and coadjoint actions, and of Lie brackets, remain 
globally the same as in the case of finite-dimensional Lie groups, and the proofs are 
similar: we only need to replace charts by plots of the underlying diffeologies (see e.g. 
\cite{Les} for further details, and \cite{BN2005} for the case of Fr\"olicher Lie groups), 
as soon as we have checked that the Lie algebra $\mathfrak{g}$ is a diffeological Lie 
algebra, i.e. a diffeological vector space equipped with a smooth Lie bracket.

\begin{Definition} \label{reg1} 
    A Fr\"olicher Lie group $G$ with Lie algebra $\mathfrak{g}$
    is called \textbf{regular} if and only if there is a smooth map \[
    Exp:C^{\infty}([0,1],\mathfrak{g})\rightarrow C^{\infty}([0,1],G)\]
    such that $g(t)=Exp(v(t))$ is the unique solution
    of the differential equation 
\begin{equation}
    \label{loga}
    \left\{ \begin{array}{rcl}
    g(0) & = & e\\  \displaystyle  
    \frac{dg(t)}{dt}\,g(t)^{-1} & = & v(t)\; ,\quad \quad
    v(t) \in C^{\infty}([0,1],\mathfrak{g}) \; .  \end{array}\right. 
\end{equation}
    We define the exponential function as follows:
    \begin{eqnarray*}
        exp:\mathfrak{g} & \rightarrow & G\\
        v & \mapsto & exp(v)=g(1) \; ,
    \end{eqnarray*}
    where $g$ is the image by $Exp$ of the constant path $v.$ \end{Definition}

We can also define the Riemann integral of smooth $\mathfrak{g}-$valued functions, 
see \cite{Ma2015}.

\begin{Definition} \label{reg2}
    Let $(V,\F, \C)$ be a Fr\"olicher vector space, i.e. a vector space $V$ 
    equipped with a Fr\"olicher structure compatible
    with  vector space addition and  scalar multiplication. The space $(V,\F, 
    \C)$ is \textbf{regular} if there is a smooth map
    $$ \int_0^{(.)} : C^\infty([0;1];V) \rightarrow C^\infty([0;1],V)$$ such 
    that $\int_0^{(.)}v = u$ if and only if $u$ is the unique solution of
    the differential equation
    \[
    \left\{ \begin{array}{l}
    u(0)=0\\
    u'(t)=v(t)\; . \end{array}\right. \]
\end{Definition}

\noindent This definition applies, for instance, if $V$ is a complete locally convex 
topological vector space, equipped with its natural Fr\"olicher structure 
given by the Fr\"olicher completion of its n\'ebuleuse diffeology, see 
\cite{Igdiff,Ma,Ma2013}.

\begin{Definition}
    Let $G$ be a Fr\"olicher Lie group with Lie algebra $\mathfrak{g}.$ Then, 
    $G$ is \textbf{fully regular}
    if both $G$ and $\mathfrak{g}$ are regular in the sense of definitions 
    \ref{reg1} and \ref{reg2} respectively.
\end{Definition}

For completeness, let us mention that ---following terminology 
used in the early investigations on infinite dimensional Lie theory 
(\cite{Om}; see also \cite{Neeb2007})--- a regular 
Lie algebra $\mathfrak{g}$ is said to be \textbf{enlargeable}
if there exists a (not necessarily regular) Fr\"olicher Lie group $ G$
with Lie algebra   $\mathfrak{g}.$

\begin{Theorem} \cite{Ma2013}
    \label{Lie}
    Let $G$ be a regular Fr\"olicher Lie group with Lie algebra $\mathfrak{g}.$
    Let $\mathfrak{g}_1$ be a Lie subalgebra of $\mathfrak{g}$, and
    set $G_1 = Exp(C^\infty([0;1];\mathfrak{g}_1))(1).$
    If $Ad_{G_1 \cup G_1^{-1}}(\mathfrak{g_1})=\mathfrak{ g}_1,$
    $$i.e. \quad \forall g \in Exp(C^\infty([0;1];\mathfrak{g}_1))(1),\quad  
    \forall v \in \mathfrak{g}_1, \quad Ad_gv \in \mathfrak{g}_1 \hbox{ and } 
    Ad_{g^{-1}}v \in \mathfrak{g}_1\; ,$$ then $G_1$ is a Fr\"olicher 
    subgroup of $G.$
\end{Theorem}

We finish this subsection stating two further key results from \cite{Ma2013,Ma2015} which 
we need in Section \ref{S3}.

\begin{Theorem} \label{regulardeformation}
    Let $(A_n)_{n \in \N^*} $ be a sequence of  (Fr\"olicher)
    vector spaces which are regular,
    equipped with a graded smooth multiplication operation
    on $ \bigoplus_{n \in \N^*} A_n ,$ i.e. a multiplication such that for each $n,m
    \in \N^*$,
    $A_n .A_m \subset A_{n+m}$ is smooth with respect to the corresponding Fr\"olicher 
    structures.
    Let us define the (non unital) algebra of formal series
    $$\A= \left\{ \sum_{n \in \N^*} a_n \,|\, \forall n \in \N^* , a_n \in A_n \right\},$$
    equipped with the Fr\"olicher structure of an infinite product. 
    Then, the set
$$
1 + \A = \left\{ 1 + \sum_{n \in \N^*} a_n | \forall n \in \N^* , a_n \in A_n \right\} 
$$
    is a Fr\"olicher Lie group with regular  Fr\"olicher Lie algebra $\A.$ 
    Moreover, the exponential map defines a smooth bijection $\A \rightarrow 1+\A.$
\end{Theorem}

\noindent
\textbf{Notation:} for each $u \in \A,$ we note by $[u]_n$ the $A_n$-component of $u.$

\begin{Theorem}\label{exactsequence}
Let
$$ 
1 \longrightarrow K \stackrel{i}{\longrightarrow} G \stackrel{p}{\longrightarrow}  H 
\longrightarrow 1 
$$
be an exact sequence of Fr\"olicher Lie groups, such that there is a smooth section 
$s : H \rightarrow G,$ and such that the subset diffeology  from $G$ on $i(K)$ coincides 
with the push-forward diffeology from $K$ to $i(K).$    We let 
$$ 
0 \longrightarrow \mathfrak{k} \stackrel{i'}{\longrightarrow} \mathfrak{g} 
\stackrel{p}{\longrightarrow}  \mathfrak{h} \longrightarrow 0  
$$
be the corresponding sequence of Lie algebras. Then,
\begin{itemize}
\item The Lie algebras $\mathfrak{k}$ and $\mathfrak{h}$ are regular if and only if the
Lie algebra $\mathfrak{g}$ is regular;
\item The Fr\"olicher Lie groups $K$ and $H$ are regular if and only if the Fr\"olicher 
Lie group $G$ is regular.
\end{itemize}
\end{Theorem}

\subsection{Existence of non regular Lie groups}
It is shown in \cite{Ma2018-2} that the Fr\"olicher
Lie group $Diff_+(]0;1[)$ of all increasing (Fr\"olicher) diffeomorphisms of the open interval, 
is non regular in the sense of definition \ref{reg1}.
The main feature of the proof is that the unit vector field $1_{ ]0;1[}$
is in the Lie algebra of $Diff_+(]0;1[)$, and that it generates a local flow
acting by translations. This local flow cannot be global.

After this example, we consider a question raised in
\cite{KM} on the existence of non-regular Lie groups. There are two
main definitions of regular Lie groups. The first one was given by
Omori \cite{Om,OMYK2} in the context of "generalized Lie groups'', which are 
topological groups without atlas but with additional structures. Omori's definition 
assumes that equation (\ref{loga}) can be solved by a convergent Euler method (details 
are given below, see Definition \ref{omo}). The second one was given by
\cite{KM}; it does not assume convergence of the Euler numerical scheme,  
the only assumption is the {\em existence} of a solution to (\ref{loga}), %without any emphasis on the method used
%to solve this equation, 
and the groups considered are ``convenient" $c^\infty-$Lie groups. Definition \ref{reg1} used for Fr\"olicher
Lie groups is an adaptation of the one given in \cite{KM}.

\vskip 12pt
\centerline{\it Are there non-regular Lie groups in the sense of Omori ?}
\vskip 12pt

Let us comment more on this question. There are actually many groups for which we cannot prove existence 
of the exponential map, even in the extended sense of \cite{KM} or Definition \ref{reg1}. 
For example, the Fr\"olicher Lie group $Diff_+(\,]0,1[\,),$ which is not a Lie group 
because there is no known atlas on it, is not regular. On the other hand, there are 
interesting groups which are regular, for example the group of  units of a (``nice'') 
algebra, modelled on a complete, Mackey complete, locally convex topological vector space, 
see \cite{Glo}: the ``minimal reasonable setting'' for infinite 
dimensional Lie groups is indeed difficult to determine \cite{Neeb2007}.  

Let us now prove that there exists a Lie group modelled on a locally convex 
topological vector space, which is not Omori-regular. We consider Omori's 
definition of a regular Lie group in a specialized setting. 

\begin{Definition} \cite{Om} \label{omo}
    Let $\A$ be a topological algebra, for which $\A^*$ is open in $\A$ and 
    it is a Lie group modelled on $\A .$
    Then $\A^*$ is Omori-regular if and only if, for any $v \in C^
    \infty([0,1],\A) ,$ there is a unique smooth path $g \in C^\infty([0,1],
    \A^*) ,$
    such that
    \begin{enumerate}
        \item $g^{-1}dg = v$ (equation for the right logarithmic derivative)
        \item $g(0)=1$ (initial condition)
        \item $g(t) = \lim_{n \rightarrow +\infty}\prod_{i=0}^{n-1}\left(1 + 
        \frac{t}{n}v\left(\frac{i}{n}\right)\right)$ (convergence of the 
        Euler method)
    \end{enumerate}
\end{Definition}

\smallskip

Let us consider 
$$
\mathbb{R}((X)) = 
\left\{ \sum_{n \in \mathbb{Z}} a_n X^n | 
\quad \exists N \in \mathbb{Z}, \quad \forall n < N, a_n = 0 \right\}\; .
$$
The space $\R((X))$ is a vector subspace of $\R^\Z,$ where the latter is equipped with 
its (product) Fr\'echet topology. $\R((X))$ is a well-known (commutative) field, and hence 
$\R((X))^* = \R((X))-\{0\}.$

%We prove the following result:

\begin{Proposition} \label{series1}
The set of invertible elements $\mathbb{R}((X))^*$ is an open subset of $\R((X)), $ and 
multiplication and inversion are smooth. As a consequence, $\mathbb{R}((X))^*$ is a Lie 
group modelled on a locally convex topological vector space.
\end{Proposition}
\begin{proof}
Since $\mathbb{R}((X))^* = \mathbb{R}((X))-\{0\},$ it is open in $\mathbb{R}((X)),$ and 
hence it is a smooth submanifold of $\mathbb{R}((X)).$ Let us now consider multiplication. 
Let $$\left(\sum_{n \in \Z} a_n X^n; \sum_{n \in \Z} b_n X^n \right)\in \R((X))^2\; ,$$ 
and let $\sum_{n \in \Z} c_n X^n$ be the product Laurent series. We have 
$$
\forall n \in \Z, \quad c_n = \sum_{i+j=n} a_i b_j\; ,
$$
and this sum is well-defined since it has only a finite number of non-zero terms. 
Moreover, for  fixed indexes $i$ and $j$, the map $(a_i,b_j)\mapsto a_ib_j$ is smooth. So 
that, for the product Fr\"olicher structure, multiplication is smooth.

Let us now consider inversion. Let $\sum_{n \geq N} a_n X^n \in \R((X))$, where 
$$
N = \min \{ n \in \Z | a_n \neq 0\}\; .
$$
Let $\sum_{n \geq -N} b_nX^n$ the inverse series. The coefficients of the inverse series 
can be calculated by induction with the formula:
    $$ \sum_{i+j=n} a_i b_j = \delta_{n,0} \quad \hbox{ (Kronecker symbol)\; . }$$
    First, we get $b_{-N} = a_{N}^{-1}$
    and by induction on $p\in \N^*,$
    $$ b_{-N+p} = -a_{N}^{-1}\sum_{k = 0}^{p-1} a_{N+p-k}b_{-N+k},$$
    which shows that inversion is smooth by the same arguments as above.
\end{proof}

\begin{Lemma} \label{series2}
    The sequence $\left(1 + \frac{X^{-1}}{n}\right)^n$ has no limit in 
    $\R((X)).$
\end{Lemma}
\begin{proof}
The Fr\"olicher structure of $\R((X))$ makes all the index-wise projections 
$[.]_m$ on the $m-$th coefficient of the Laurent series smooth. 
We remark that for any $m>0$, 
$$\left[\left(1 + \frac{X^{-1}}{n}\right)^n\right]_m = 0$$
    and that, for $m \in \N, $
    $$\lim_{n \rightarrow + \infty} 
    \left[\left(1 + \frac{X^{-1}}{n}\right)^n\right]_{-m} = 
    \lim_{n \rightarrow +\infty} 
    \frac{1}{n^m}\left(\begin{array}{c} n \\ m \end{array} \right) = 
    \frac{1}{m!} \neq 0.$$
    Thus this sequence converges in $\R^\Z,$ but not in $\R((X)).$
\end{proof}

\noindent These results imply the following theorem:

\begin{Theorem} \label{series3}
    $\mathbb{R}((X))^*$ is not regular in the sense of Omori.
\end{Theorem}
\begin{proof}
    We have that $X^{-1} \in \R((X))$ so that the constant path
    $$ t \in [0,1] \mapsto v(t)=X^{-1} \in C^\infty\left(\,[0,1],\R((X))\,
    \right).$$
    The Euler method gives:
    $$\prod_{i=0}^{n-1}\left(1 + \frac{t}{n}v\left(\frac{i}{n}\right)\right) 
    = \left(1 + \frac{t X^{-1}}{n}\right)^n.$$
    By Lemma \ref{series2}, this sequence does not converge for $t=1,$ and 
    hence $\R((X))$ is not regular in the sense of Omori.
\end{proof}

\subsection{Principal bundles, connections and Ambrose-Singer theorem}

Let $P$ be a diffeological space and  let $G$ be a Fr\"olicher Lie
group, with a smooth right-action $P\times G \rightarrow P,$ such
that $\forall (p,p',g)\in P\times P \times G,$ we have $p.g=p'.g
\Rightarrow p=p'$. Let $M=P/G,$ equipped with the quotient
diffeology. $P$ is called a \textbf{principal $G-$bundle} with base
$M.$ In \cite[Article 8.32]{Igdiff} Iglesias-Zemmour gives a
definition of a connection on a principal $G$-bundle in terms of
paths on the total space $P,$ generalising the classical notion 
of path lifting for principal bundles with finite-dimensional
Lie groups as structure groups.

\begin{Definition}\label{d:iz-connection}
    Let $G$ be a diffeological group, and let
    $\pi\colon P\to X$ be a principal $G$-bundle.
    Denote by $\pathloc(P)$ the diffeological space of
    \textbf{local paths} (see \cite[Article 1.63]{Igdiff}), and by 
    $tpath(P)$ the \textbf{tautological bundle of local paths} 
    $$tpath(P):=\{(\gamma,t)\in\pathloc(P)\times\R\mid t\in D(\gamma)\}.$$ 
    A \textbf{diffeological connection} is a smooth map 
 $H \colon tpath(P)\to\pathloc(P)$ satisfying the following properties
     for any $(\gamma,t_0)\in tpath(P)$:
    \begin{enumerate}
        \item the domain of $\gamma$ equals the domain of $H(\gamma,t_0)$,
        \item $\pi\circ\gamma=\pi\circ H(\gamma,t_0)$,
        \item $H(\gamma,t_0)(t_0)=\gamma(t_0)$,
        \item $H(\gamma\cdot g,t_0)=H(\gamma,t_0)\cdot g$ for all $g\in G$,
        \item $H(\gamma\circ f,s)=H(\gamma,f(s))\circ f$ for any smooth map
         $f$ from an open subset of $\R$ into $D(\gamma)$,
        \item $H(H(\gamma,t_0),t_0)=H(\gamma,t_0)$.
    \end{enumerate}
\end{Definition}

Another formulation of this definition can be found in \cite{Ma2013} under the terminology 
of \textbf{path-lifting}.

\begin{rem}\label{r:iz-connection}
    Diffeological connections satisfy many of the usual properties that 
    classical connections on a principal $G$-bundle (where $G$ is a finite-
    dimensional Lie group) enjoy; in particular, they admit unique horizontal 
    lifts of paths in $\pathloc(M)$ 
    \cite[Article 8.32]{Igdiff}, and they pull back by smooth maps 
    \cite[Article 8.33]{Igdiff}.
\end{rem}

\begin{Proposition}
  Let $V$ be a vector space. Then, $G$ acts smoothly from the right on 
    the space $\Omega(P,V)$ of $V$-valued differential forms on $P$ by setting
$$ 
\forall (g,\alpha) \in  \Omega^n(P,V) \times G, \forall p\in \p(P), \quad 
(g_*\alpha)_{g.p} = \alpha_p \circ (dg^{-1})^n \; .
$$
\end{Proposition}
\begin{proof}
$G$ acts smoothly on $P$ so that, if $p \in \p(P),$ $g.p \in \p(P)$.
The right action is now well-defined, and smoothness is trivial.
\end{proof}

\begin{Definition}
    Let $\alpha \in \Omega(P;\mathfrak{g}).$ The differential form $\alpha$ is 
    \textbf{right-invariant} if and only if, for each $p \in \p(P),$ and for each 
    $g \in G,$
    $$\alpha_{g.p} = Ad_{g^{-1}} \circ g_*\alpha_p \; .$$
\end{Definition}

Now, let us turn to connections and holonomy. Let $p \in P$ and
let $\gamma$ be a smooth path in $P$ starting at $p.$

\begin{Definition}
    A \textbf{connection} on $P$ is a $\mathfrak{g}-$valued right-invariant 
    $1$-form $\theta,$ , such that, for each $ v \in \mathfrak{g},$ for any path 
    $c : \R \rightarrow G$ such that 
$$
\left\{\begin{array}{ccr} c(0)& = & e_G\\ \d_tc(t)|_{t=0}&= & v \; \; ,\end{array} \right.
$$ 
and for each $p \in P$ we have: 
    $$\theta(\d_t(p.c(t))_{t = 0})=v \; .$$
\end{Definition}

Now, let $p \in P$ and
$\gamma$ a smooth path in $P$ starting at $p,$ defined on $[0,1].$
Let $H_\theta \gamma (t) = \gamma(t)g(t)$, where $g(t) \in C^\infty([0,1];\mathfrak{g})$ 
is a path satisfying the differential equation:
$$
\left\{ \begin{array}{c} \theta \left( \d_t H_\theta\gamma(t) \right) = 0  \\ 
H_\theta\gamma(0)=\gamma(0) \end{array} \right.
$$
The first line of this equation is equivalent to the differential
equation $$g^{-1}(t)\d_tg(t) = -\theta(\d_t\gamma  (t))$$ which is
integrable, and the second line is equivalent to the initial condition $g(0)=e_G.$
This shows that horizontal lifts are well-defined, as in the standard case
of finite-dimensional manifolds. Moreover, the map $H_\theta(.)$ defines trivially a
diffeological connection. This enables us to consider the holonomy
group of the connection. Notice that a straightforward adaptation
of the arguments of \cite{Ma} shows that the holonomy group does not depend (up to 
conjugation and up to the choice of connected component of $M$) on the choice of the base 
point $p.$ Now we assume that $dim(M)\geq 2$ and we fix a connection
$\theta$ on $P.$

\begin{Definition}
Let $\alpha \in \Omega(P;\mathfrak{g})$ be a $G-$invariant $1$-form. Let 
$\nabla \alpha = d\alpha - {\frac{1}{ 2}}[\theta,\alpha]$ be the horizontal derivative 
of $\alpha.$ The curvature $2$-form induced by $\theta$ is 
 $$ \Omega = \nabla \theta \; .$$ 
\end{Definition}

This definition allows us to consider reductions of the structure group.

 \begin{Theorem} \label{Courbure} \cite{Ma2013}
    We assume that $G_1$ and $G$ are regular Fr\"olicher groups
    with regular Lie algebras $\mathfrak{g}_1$ and $\mathfrak{g}.$
    Let $\rho: G_1 \mapsto G$ be an injective morphism of Lie groups.
    If there exists a connection $\theta$ on $P$, with curvature $\Omega$, 
    such that for any smooth $1$-parameter family $H_\theta c_t$ of 
    horizontal paths starting at $p$, and for any smooth vector fields $X,Y$ 
    in $M$, the map
    \begin{eqnarray} 
    s, t \in [0,1]^2 & \rightarrow & \Omega_{Hc_t(s)}(X,Y)  \label{g1}
    \end{eqnarray}
    is a smooth $\mathfrak g_1$-valued map (for the $\mathfrak g _1 -$ 
    diffeology),
    \noindent
    and if $M$ is simply connected, then the structure group $G$ of $P$ reduces to 
    $G_1,$ and the connection $\theta$ also reduces.
 \end{Theorem}

We can now state the announced Ambrose-Singer theorem, using the terminology of \cite{Rob} 
for the classification of groups via properties of the exponential map:

 \begin{Theorem}
    \label{Ambrose-Singer} \cite{Ma2013}
Let $P$ be a principal bundle whose  structure group is a fully regular Fr\"olicher 
Lie group $G$. Let $\theta$ be a connection on $P$ and $H_\theta$ the associated 
diffeological connection.
    \begin{enumerate}
        \item For each $p \in P,$ the holonomy group $\Hol_p^L$ is a
        diffeological subgroup of $G$, which does not depend on the choice of
        $p$ up to conjugation.

        \item There exists a second holonomy group $H^{red},$ $\Hol \subset H^{red},$
        which is the smallest structure group for which there is a subbundle $P'$ to
        which $\theta$ reduces. Its Lie algebra is spanned by the curvature elements, i.e.
        it is the smallest integrable Lie algebra which contains the 
        curvature elements.

        \item If $G$ is a Lie group (in the classical sense) of type I or II,
        there is a (minimal) closed Lie subgroup $\bar{H}^{red}$ (in the 
        classical sense) such that $H^{red}\subset \bar{H}^{red},$
        whose Lie algebra is the closure in $\mathfrak{g}$ of the Lie algebra 
        of $H^{red}.$ $\bar{H}^{red}$
        is the smallest closed Lie subgroup of $G$ among the structure groups
        of closed sub-bundles $\bar{P}'$ of $P$ to which $\theta$ reduces.
    \end{enumerate}
 \end{Theorem}

 From \cite{Ma2013} again, we have the following result:

 \begin{Proposition} \label{0-courbure}
    If the connection $\theta$ is flat and $M$ is connected and simply 
    connected, then for any path $\gamma$
    starting at $p \in P,$ the map $$\gamma \mapsto H_\theta\gamma(1)$$ 
    depends only on $\pi(\gamma(1))\in M$, and it defines
    a global smooth section $M \rightarrow P.$ Therefore, $P = M \times G.$
 \end{Proposition}

 Let us precise a little bit more this result (see \cite[section 40.2]{KM} 
 for an analogous statement in the $c^\infty$-setting):

\begin{Theorem} \label{Hslice}
    Let $(G,\mathfrak{g})$ be a regular Lie group with regular Lie algebra 
    and let $X$ be a simply connected
    Fr\"olicher space. Let $\alpha \in \Omega^1(M,\mathfrak{g})$ such that 
    \begin{equation} \label{beta1}
    d\alpha + [\alpha,\alpha]=0\; .
    \end{equation}
    Then there exists a smooth map $$f : X \rightarrow G $$
    such that $$df.f^{-1} = \alpha.$$ Moreover, we move from one solution 
    $f$ to another by applying the
    Adjoint action of $G$, pointwise in $x \in X$.
\end{Theorem}

We remark that the theorem also holds if we consider the equation
$$d\alpha - [\alpha,\alpha]=0$$ instead of (\ref{beta1}); we only need to 
change left logarithmic derivatives for right logarithmic derivatives,
and Adjoint action for Coadjoint action. The correspondence
between solutions is given by the inverse map $f \mapsto f^{-1}$
on the group $C^\infty(X,G).$

\section{On groups and algebras of formal series and pseudo-differential 
operators } \label{S3}

\subsection{The algebra of formal pseudo-differential operators}\label{3.1}

We let $A$ be a commutative $\K-$algebra with unit $1$, in which $\K$ is any topologically
complete field of characteristic zero. Let $A^*$ be the group of units of $A,$ {\em i.e.},
the group of invertible elements of $A.$ We assume that $A$ is equipped with a derivation,
that is, with a $\K$-linear map $\d : A \rightarrow A$
satisfying the Leibnitz rule $\d (f \cdot g) = (\d f) \cdot g + f \cdot (\d g)$
for all $f,g \in A$. 

Let $\xi$ be a formal variable not in $A$. The {\em algebra of
    symbols} over $A$ is the vector space
\[
\Psi_{\xi}(A) = \left \{ P_{\xi} = \sum_{\nu \in {\bf Z}} a_{\nu} \, \xi^{\nu} : a_{\nu}
\in A \, , \;  a_{\nu} = 0 \mbox{ for } \nu \gg 0 \right \} \;
\]
equipped with the associative multiplication $\circ\,$ given by
\begin{equation}  \label{alfa}
P_{\xi} \circ Q_{\xi} = \sum_{k \geq 0} \frac{1}{k !} \, \frac{\partial^{k}P_{\xi}}
{\partial \xi^{k}} \, \d^{k} Q_{\xi} \; ,
\end{equation}
with the prescription that multiplication on the right hand
side of (\ref{alfa}) is standard multiplication of Laurent series in $\xi$
with coefficients in $A$, see \cite{D,O}. The algebra $A$ is
included in $\Psi_\xi(A)$. Operation (\ref{alfa}) mirrors the extension of the Leibnitz 
rule to negative powers of the derivative $\partial$ as explained, for example, in 
\cite{O}.

The {\em algebra of formal pseudo-differential operators} over $A$ is the vector space
\[
\Psi (A) = \left \{ P = \sum_{\nu \in \mathbb{Z}} a_{\nu} \,
\partial^{\nu} : a_{\nu} \in A \, , \; a_{\nu} = 0 \mbox{ for } \nu \gg 0 \right \}
\]
equipped with the unique multiplication which makes the map
$\sum_{\nu \in \mathbb{Z}} a_{\nu} \, \xi^{\nu} \; \mapsto \;
\sum_{\nu \in \mathbb{Z}} a_{\nu} \, \partial^{\nu}$
an algebra homomorphism. The
algebra $\Psi (A)$ is associative but not commutative. It becomes
a Lie algebra over $K$ if we define, as usual,
\begin{equation}
[ P , Q ] = P \, Q - Q \, P \; ,  \label{liebrack}
\end{equation}
and it is also an example of an infinite dimensional non-commutative
Poisson algebra in the sense of Kubo, see \cite{Ku}.

The {\em order} of $P \neq 0 \in \Psi (A)$, $P = \sum_{\nu \in \mathbb{Z}} a_{\nu} \,
\partial^{\nu}$, is $N$ if $a_{N} \neq
0$ and $a_{\nu} = 0$ for all $\nu > N$. If $P$ is of order $N$, the coefficient
$a_{N}$ is called the {\em leading term} or {\em principal symbol} of $P$. We note by 
$\Psi^N(A)$ the vector space of pseudo-differential operators $P$ as above satisfying
$k > N \Rightarrow a_k = 0,$  that is,
$$
\Psi^N(A)=\left\{P \in \Psi(A)\,|\, {\rm order}(P) \leq N \right\}\; .
$$ 
The vector space $\Psi^0(A)$ is of particular interest, since direct computations show 
that it is an algebra. We note by $\Psi^{0,*}(A)$ its group of units, i.e. the group of 
invertible elements of $\Psi^0(A).$ We collect some classical properties of $\Psi(A)$.

\begin{Lemma} \label{lieprop}

~

\begin{enumerate}
\item If $P$ and $Q$ are formal pseudo-differential operators over $A$, and the leading 
terms of $P$ and $Q$ are not divisors of zero in $A$, then 
$$
{\rm order}(P \, Q) = {\rm order}(P) + {\rm order}(Q)\; ,
$$ 
and
$$ 
{\rm order}([P,Q]) \leq {\rm order}(P) \, + {\rm order}(Q) -1\; .
$$ 
\item Every non-zero formal pseudo-differential operator $P$ for which its leading term
      is invertible has an inverse in $\Psi (A)$.
    \end{enumerate}
\end{Lemma}

This lemma is proven for instance in \cite{O}. After \cite{Ma2015}, we assume now, and 
till the end of subsection \ref{3.1}, that the algebra $A$ is a \textbf{Fr\"olicher 
algebra}. Thus, addition, scalar multiplication and multiplication are smooth, and 
inversion is a smooth operation on $A^*$, in which $A^*$ is equipped with the subset 
Fr\"olicher structure. We also assume that $\partial$ is smooth.
Then, identifying a formal pseudo-differential operator 
$P \in \Psi(A)$ with its sequence of partial symbols, we obtain that  
$\Psi(A),$ as a linear subspace of $A^\Z,$ carries a natural Fr\"olicher structure. 
We obtain the following special case of \cite[Proposition 3.25]{Ma2015}:

\begin{Proposition} \label{Psi-F}
    $\Psi(A)$ is a Fr\"olicher algebra.
\end{Proposition}

\begin{rem}
In more classical settings, $\Psi(A)$ can have a quite complicated structure. For example, 
if $A$ is a locally convex algebra, so is $\Psi(A)$, but this space is not
complete for the product topology even if $A$ is complete, see \cite[Section 5]{ARS1}.
\end{rem}

We further assume, until the end of subsection \ref{3.1}, that $A^*$ is a 
Fr\"olicher Lie group with Lie algebra $\mathfrak{g}_A.$

\smallskip

\smallskip

We notice that if $A$ is a complete locally convex topological algebra equipped with its 
natural Fr\"olicher structure, and $A^*$ is open in $A,$ then $\mathfrak{g}_A = A.$
For the needs of integration, we also have to assume (till the end of subsection 
\ref{3.1}) that $A$ is regular as a Fr\"olicher vector space.
We now extend a result proved in \cite{Glo} in the case when $A$ is a Fr\'echet algebra 
and the group of units is open in $A$. In that paper, Gl\"ockner was able to adapt the 
standard proof of the Lie group structure of $GL(H)$ when $H$ is a Hilbert space to his
Fr\'echet algebra context; since in our setting $A$ can be more general, our proof needs 
to be different. It is based on Theorem \ref{regulardeformation} and Theorem 
\ref{exactsequence}.

\begin{Lemma} \label{Psi-1}
The group $1+ \Psi^{-1}(A)$ is a fully regular Fr\"olicher Lie group with regular Lie 
algebra $\Psi^{-1}(A).$
\end{Lemma}
\begin{proof}
By integration componentwise, we already know that $ \Psi^{-1}(A)$ is regular.  We now 
remark that $\Psi^{-1}(A)$ is graded by the order. Then, setting 
$$
\A_n = \{a_{-n}\d^{-n}| a_{-n}\in A\}\; ,
$$ 
we obtain the hypotheses needed in order to apply Theorem \ref{exactsequence}.
\end{proof}

\begin{Theorem}
There exists a short exact sequence of groups 
$$ 
1 \longrightarrow 1 + \Psi^{-1}(A) \longrightarrow \Psi^{0,*}(A) \longrightarrow 
A^* \longrightarrow 1
$$
such that:
\begin{enumerate}
\item The injection $1 + \Psi^{-1}(A) \rightarrow \Psi^{0,*}(A)$ is smooth
\item The principal symbol map $ \sigma_0:\Psi^{0,*}(A) \rightarrow A^*$ is smooth and it 
has a global section which is the restriction to $A^*$ of the canonical inclusion 
$A \rightarrow \Psi^{0}(A).$
\end{enumerate}
As a consequence, $A^*$ is a fully regular Fr\"olicher Lie group if and only if  
$\Psi^{0,*}(A)$ is a fully regular Fr\"olicher Lie group with regular Lie algebra 
$\mathfrak{g}_A \oplus \Psi^{-1}(A).$
\end{Theorem}
\begin{proof}
    Let $a,b \in (\Psi^0(A))^2.$ Then, $\sigma_0(ab) = \sigma_0(a) \sigma_0(b).$ If $a$ 
    is invertible, setting $b = a^{-1},$ we get $\sigma_0(a)\sigma_0(a^{-1}) = 1_A$. 
    Thus, $\sigma_0(a) \in A^*.$ Thus, the map $a \mapsto \sigma_0(a)$ is a
    morphism of diffeological groups. Moreover, the section map
    $A^* \rightarrow \Psi^{0,*}(A)$ is a smooth morphism of Fr\"olicher Lie groups, and 
    the inclusion $1+ \Psi^{-1}(A) \rightarrow \Psi^{0,*}(A)$ is obviously a morphism of 
    Fr\"olicher Lie groups. Lemma \ref{Psi-1} implies
    that we can apply Theorem \ref{exactsequence}, thereby completing the proof.
\end{proof}

\vskip 8pt

As explained in Remark \ref{comp}, the results stated above
remain valid for classical G\^ateaux differentiability if $A$ is a
Fr\'echet space.

\smallskip

\smallskip

Let us now prove the following: 

\begin{Theorem}
There is no Fr\"olicher subalgebra $\A$ of $\Psi(A)$ whose group of units is a 
Fr\"olicher Lie subgroup of $\Psi(A)^*$ with Lie algebra $\mathfrak{g}$ and 
$\d \in \mathfrak{g},$ which is Omori-regular. In particular $\Psi(A)^*$ is not 
Omori-regular.
\end{Theorem}
\begin{proof}
    Substituting $X = \d^{-1}$ in Lemma \ref{series2} we have that the sequence 
    $$u_n = \left(1 + \frac{\d}{n}\right)^n$$
    converges in $\R^\Z$ but not in $\R((\d^{-1})) = \R^\Z \cap \Psi(A).$
\end{proof}

The fact that the Euler method fails to define an exponential map motivates us, after 
\cite{M1,M3}, to regularize $\Psi(A)$ into some ``deformed'' algebra, see 
Definitions 3.8, 3.17 and 4.3 below.

\subsection{Mulase formal Lie group and Fr\"olicher structures}

Following Bourbaki \cite[Algebra, Chapters 1-3, p. 454-457]{B}, we take $T$ as the
additive monoid of all sequences of natural numbers $t=(n_i)_{i\in \N^*}$ such that 
$n_{i}=0$ except for a finite number of indices, and $R$ as a commutative ring equipped 
with a derivation $\partial$. A {\em formal power series} is a function $u$ from $T$ to 
$R$, $u=(u_t)_{t\in T}$. Consider an infinite number of formal variables 
$\tau_i$, $i \in I$, and set $\tau=(\tau_1,\tau_2,\cdots)$. Then, the formal power series
$u$ can be written as $u=\sum_{t\in T}u_t\tau^t$ in which 
$\tau^t = \tau_1^{n_1} \tau_2^{n_2}\cdots$. We say that $u_t\in R$ is a
coefficient and that $u_t\tau^t$ is a term. 

Operations on the set $R[[\tau]] := R[[\tau_1,\tau_2, \cdots]]$ are defined 
in an usual manner: If $u=(u_t)_{t\in T},$ $v=(v_t)_{t\in T},$ then
$$
u+v=(u_t+v_t)_{t\in T} \quad \mbox{ and } \quad u\,v=w\; ,
$$
in which $w=(w_t)_{t\in T}$ and
$$
%w_t=\sum_{r,s\in T\atop{r+s=t}} u_r v_s\; .
w_t=\sum_{\genfrac{}{}{0pt}{1}{r,s\in T}{r+s=t}} u_r\, v_s\; .
$$
It is shown in \cite[p. 455]{B} that this multiplication is well defined, and that 
$A_t = R[[\tau]]$ equipped with these two operations is a commutative 
algebra with unit. The derivation $\partial$ on $R$ extends to a derivation
on $A_t$ via 
$$
\d u=\sum_{t\in T}( \d u_t)\tau^t\; .
$$

We define the \textit{valuation} of a power series following \cite{M1}, p. 59. Let 
$u\in A_t,$ $u\neq 0.$ We write $u=\sum_{t\in T}u_t\tau^t$, and if 
$t = (n_i)_{i \in \N}$, we set $\mid t\mid= \sum i n_i$, which is equivalent to Mulase's 
assumption $ord (\tau_i) = i.$ The terms $u_t\tau^t$ such that $\mid t\mid=p$ are called 
{\em monomials of valuation} $p.$ The formal power series $u_p$ whose terms of
valuation $p$ are those of $u,$ and whose other terms are zero,
    is called {\em the homogeneous part of $u$ of valuation $p.$} The
    series $u_0,$ the homogeneous part of $u$ of valuation $0$, is
    identified with an element of $R$ called the constant term of $u.$
    For a formal series $u\neq0,$ the least integer $p\geqslant0$ such
    that $u_p\neq0$ is called the {\em valuation} of $u,$ and it is
    denoted by $val_t(u)$. We extend this definition to the case
    $u=0$ by setting $val_t(0) =+ \infty$ (see \cite[p. 457]{B}). The
    following properties hold: if $u,v$ are formal power series then
\begin{equation} \label{val1} 
val_t(u+v)\geqslant inf(val_t(u),val_t(v))\; , \quad \mbox{ if } u+ v \neq 0 \; , 
\end{equation}
\begin{equation} \label{val2} 
val_t(u+v)=inf(val_t(u),val_t(v))\; , \quad \mbox{ if } val_t(u)\neq val_t(v) \; ,
\end{equation}
\begin{equation} \label{val3}
val_t(uv)\geqslant val_t(u)+val_t(v)\; , \quad \mbox{ if } u v \neq 0\; .
\end{equation}

\smallskip

\begin{rem}
Other definitions of $val_t$ are possible, see \cite{D1,D2} and \cite{ER2013}.
\end{rem}

\begin{Definition}[Mulase, \cite{M3}]
We define spaces of formal pseudo-differential and differential  operators of infinite 
order, $\widehat{\Psi}(A_t)$ and $\widehat{\D}_{A_t}$ respectively, as
    \smallskip
\begin{equation} \label{aldef}
\widehat{\Psi}(A_t) = 
     \left\{ \sum_{\alpha \in {\mathbb{Z}}}
    a_{\alpha}\,\d^{\alpha} : a_{\alpha} \in A_t \mbox { and } \exists
     (C,N)\in\R \times \mathbb{N} \mbox{ such that }
    val_t(a_\alpha) > C\alpha - N \ \forall\ \alpha \gg 0 \right\}
\end{equation}
and
\begin{equation} \label{aldef1}
    \widehat{\D}_{A_t} = \left\{ P= \sum_{\alpha \in \mathbb{Z}}
    a_{\alpha}\,\d^{\alpha} :
     P \in\widehat{\Psi}(A_t) \mbox{ and }
    a_\alpha=0 \mbox { for } \alpha <0 \right\} \; .
\end{equation}
\end{Definition}

\noindent In addition, we define

$$
\mathcal{I}_{A_t} = 
\Psi^{-1}(A_t)=\left\{P \in\widehat{\Psi}(A_t) : \forall \alpha > 0, a_\alpha = 0\right\}
$$
and $$G_{A_t}= 1 + \mathcal{I}(A_t)\; .$$
Notice that
$A_t \subset \widehat{\Psi}(A_t)$ as the 0-order term,  
$$ G_{A_t} \subset \Psi(A_t) \subset\widehat{\Psi}(A_t)\; ,$$ and 
$\widehat{\D}_{A_t} \not \subset \Psi(A_t).$ The operations on $\widehat{\Psi}(A_t)$ are
extensions of the operations on $\Psi(A_t)$. More precisely we have (\cite{M3}, see also 
an explicit proof in \cite{ER2013}):

\begin{Lemma}
    The space $\widehat{\Psi}(A_t)$ has an algebra structure, and
    $\widehat{\D}_{A_t}$ is a subalgebra of $\widehat{\Psi}(A_t)$.  
\end{Lemma}

\begin{Definition} \label{tez}
    Let $\mathcal{K}$ be the ideal of $A_t$ generated by $\tau_1, \tau_2,
    \cdots$. If $P \in \widehat{\Psi}(A_t)$, we denote by $P\vert_{\tau=0}$ 
    the equivalence class $P\, mod\; \mathcal{K}$ (i.e. the projection on the constant 
    coefficient of the $\tau$-series), and we identify it
    with an element of $\Psi(A)$. We also define the spaces
    \begin{equation} \label{345}
    G(\widehat{\Psi}(A_t)) = \{ P \in \widehat{\Psi}(A_t) :
    P\vert_{\tau=0}\in G_{A_t}  \}
    \end{equation}
    and
    \begin{equation} \label{346}
    \widehat{\D}_{A_t}^{\times} = \{ P \in \widehat{\D}_{A_t} :
    P\vert_{\tau=0}=1 \} \; .
    \end{equation}
\end{Definition}

\smallskip

\begin{rem} \label{fring} 
We note that if $R$ is a Fr\'echet ring, then the spaces defined in 
(\ref{aldef}) and (\ref{aldef1}) are locally convex but not necessarily 
complete, and the same remark holds for the spaces defined in Lemma 
\ref{Psi-Rt} below. On the other hand, we will see
(Theorem \ref{Psihat}) that these spaces have Fr\"olicher structures and that 
therefore they can be studied in a rather complete way without the need to 
consider delicate convergence issues.
\end{rem}

Hereafter we assume that $R$ is a Fr\"olicher $\K-$algebra and that 
$\d: R \rightarrow R$ is a smooth derivation. 
These assumptions imply the following (we use the notations introduced at the 
beginning of this subsection):

\begin{Lemma} \label{Psi-Rt}
$\Psi(R)[[\tau]]$ is a Fr\"olicher algebra and 
    $$G\Psi_t(R)=\{u \in\Psi(R)[[\tau]] : \hbox{ the monomial of valuation }
    0\hbox{ equals } 1\}$$
    is a Fr\"olicher Lie group with Lie algebra
    $$\mathfrak{g}\Psi_t(R)=\{u \in\Psi(R)[[\tau]] : \hbox{ the monomial of 
    valuation }0\hbox{ equals } 0\}.$$
\end{Lemma}
\begin{proof}
By Proposition \ref{Psi-F},  $\Psi(R)$ is a Fr\"olicher algebra. Since the space of power
series $\Psi(R)[[\tau]]$ can be identified with the vector space of functions $\Psi(R)^T
\subset \Psi(R) \times T$, we conclude that $\Psi(R)[[\tau]]$ admits a natural Fr\"olicher
structure because of the results reviewed in Subsections 2.4 and 2.5. We need to show that
the algebra operations are smooth with respect to this structure. 
%        Let us look at monomials of the series $u\in \Psi(R)[[\tau]].$ 
Addition and scalar multiplication are those of the vector space $\Psi(R)^T,$ so that they 
are smooth. Now, if 
$$
u = \sum_{t \in T} a_t \tau^t \quad \mbox{ and } 
\quad v = \sum_{t \in T} b_{t} \tau^{t} \; ,
$$
then $u.v = w,$ in which $w$ is given by the formulae recalled above, when we introduced 
operations in the algebra of series $A_t$: fixing $t \in T,$ the 
coefficient in $\Psi(R)$ of $\tau^t$ is a (finite) linear combination of multiplications 
in $\Psi(R)$, namely, the products $a_{t'}\,b_{t''}$ with $t'+t''=t.$ Thus multiplication 
is smooth and so $\Psi(R)[[\tau]]$ is a Fr\"olicher algebra. 
    
Now let $u \in G\Psi_t(R).$ The coefficients of $u^{-1}$ are obtained 
from the equation $$ u u^{-1} = 1$$ which gives, by induction on the valuation in $t$ 
of each monomial:
    $$ [u^{-1}]_0 = ([u]_0)^{-1} = 1$$
    and $$ [u^{-1}]_t = - \sum_{t'+t''=t\, ,\, t'\neq 0} [u]_{t'}[u^{-1}]_{t''}$$
    which shows that inversion is smooth, since $\Psi(R)$ is a Fr\"olicher algebra.
    
Since $\Psi(R)[[\tau]]$ is a Fr\"olicher algebra, the Adjoint map
\begin{eqnarray*}
G\Psi_t(R) \times \mathfrak{g}\Psi_t(R) & \rightarrow & \mathfrak{g}\Psi_t(R) \\ 
(g,v) &\mapsto& gvg^{-1}
\end{eqnarray*}
differentiates to the classical Lie bracket
\begin{eqnarray*}
\mathfrak{g}\Psi_t(R) \times \mathfrak{g}\Psi_t(R) & \rightarrow & \Psi_t(R) \\ 
(v_1,v_2) &\mapsto& [v_1,v_2]= v_1v_2 - v_2v_1 \; ,
\end{eqnarray*}
see the discussion leading to (2.1). Now 
$$
val_t\left(v_1v_2 - v_2v_1\right) \geq inf( val_t (v_1),val_t(v_2)) \geq 1\; ,
$$ 
and hence $[v_1,v_2] \in \mathfrak{g}\Psi_t(R)$,
which ends the proof.
\end{proof}

\begin{Theorem} \label{Psihat}
    The following algebras are Fr\"olicher algebras:

\smallskip

\centerline{(1) $A_t\,$; \quad \quad (2) $\Psi(A_t)\, ;$ \quad \quad 
(3) $\widehat{\Psi}(A_t)\, ;$\quad \quad (4) $\widehat{\D}_{A_t}\, .$}
\end{Theorem} 
\begin{proof}
~
\begin{enumerate}
\item We argue as in the previous lemma. Let us look at monomials of the 
series $u\in A_t.$ Addition and scalar multiplication are those induced by the vector 
space $R^T,$ so that they are smooth. If $u.v = w,$ fixing $t \in T,$ the $R-$value of the 
$t-$monomial is a (finite) linear combination of multiplications the $R-$values of the 
$t'-$ and $t''-$monomials of $u$ and $v$ respectively, where $t'+t''=t.$ Thus 
multiplication is smooth.
\item It follows from Proposition \ref{Psi-F} and the previous item.
\item We recall that the notation $\tau^t$ makes sense because of the conventions
      at the beginning of this subsection. We follow the proof of Lemma \ref{Psi-Rt}. 
      
Let $P = \sum_{n \in \Z} a_n \partial^n \in \widehat{\Psi}(A_t)$. Then formally we have
\begin{equation} \label{id1}
\sum_{n \in \Z} a_n \partial^n =  \sum_{n \in \Z} \left( \sum_{t \in T}
u_{n t}\tau^t \right) \partial^n = 
\sum_{t \in T} \left( \sum_{n \in \Z} u_{n t} \partial^n \right) \tau^t
\end{equation} 
for 
$u_{nt} \in R$, and so we can identify $\widehat{\Psi}(A_t)$with a subalgebra of 
$\Psi(R)[[\tau]]^\Z.$ By Proposition \ref{Psi-F},  $\Psi(R)$ is a Fr\"olicher algebra, 
by Lemma \ref{Psi-Rt}, $\Psi(R)[[\tau]]$ is a Fr\"olicher algebra. Thus, by Proposition
2.15, $\widehat{\Psi}(A_t)$ is a Fr\"olicher algebra. 
     \item $\widehat{\D}_{A_t}$ is a subalgebra of $\widehat{\Psi}(A_t)$ and therefore 
     it is a Fr\"olicher algebra.
\end{enumerate}
\end{proof}

\noindent We obtain the following result at the group level:

\begin{Theorem} \label{Ghat} The following assertions hold:
~
    \begin{enumerate}
        \item   $\widehat{\D}_{A_t}^{\times}$ is a Fr\"olicher group.
        \item $G_R = 1 + \Psi^{-1}(R)$ and $G_{A_t}=1+\Psi^{-1}(A_t)$ are 
        Fr\"olicher Lie groups.
        \item $ G(\widehat{\Psi}(A_t))$ is a Fr\"olicher group.
    \end{enumerate}
\end{Theorem}
\begin{proof}

~

\begin{enumerate}
\item We proceed as in Theorem \ref{Psihat}, see Equation (\ref{id1}). 
The group $\widehat{\D}_{A_t}^{\times}$ can be identified with a subgroup of 
$G\Psi_t(R)$, on which multiplication and inversion are smooth.
\item These facts follow from Lemma \ref{Psi-1} and Theorem \ref{Psihat}.
\item Recall that we denote by $P\vert_{\tau=0}$ the projection of the coefficients of $P$ 
on their constant coefficient, see Definition \ref{tez}. Of course, this coincides with 
``evaluating at $\tau_1=\tau_2=...=0$", which explains the notation. Since 
$$P\vert_{\tau=0}(G(\widehat{\Psi}(A_t)))= G_R = 1 + \Psi^{-1}(R)$$ 
we have the following exact sequence:
$$ 0 \rightarrow Ker\left(P\vert_{\tau=0}\right)\rightarrow    
G(\widehat{\Psi}(A_t)) \rightarrow G_R \rightarrow 0$$
and there is a global section $G_R\rightarrow G(\widehat{\Psi}(A_t)),$ given by the 
canonical inclusion componentwise, which is smooth. Then adapting the classical 
(algebraic) construction of the semi-direct product 
$$G(\widehat{\Psi}(A_t)) =  Ker\left(P\vert_{\tau=0}\right) \rtimes G_R$$ 
we remark that all the necessary operations to build the multiplication and the inversion 
group structure of $G(\widehat{\Psi}(A_t))$ from $Ker\left(P\vert_{\tau=0}\right)$ and 
$G_R$ are smooth for the subset Fr\"olicher structures inherited from 
$\widehat{\Psi}(A_t).$ 
\end{enumerate}
\end{proof}

We are ready to state and prove an enriched version of Mulase's
algebraic factorization theorem, see \cite{M1,M3}, for the
Fr\"olicher group $G(\widehat{\Psi}(A_t))\,$:

\begin{Theorem} \label{mu2}

~

\begin{enumerate}
    \item   For any $U \in G(\widehat{\Psi}(A_t))$ there exist unique $S \in
    G_{A_t}$ and $Y \in \widehat{\D}_{A}^{\times}$ such that
    \[
    U = S^{-1}\,Y \; .
    \]
    In other words, there exists a unique global factorization of the 
    Fr\"olicher Lie group $G(\widehat{\Psi}(A_t))$ as a group defined by matched pairs,
    \[
    G(\widehat{\Psi}(A_t)) = G_{A_t} \, \widehat{\D}_{A_t}^\times \; .
    \]
    \item The factorization $U \mapsto (S,Y)$ is (Fr\"olicher)
    smooth. Hence, it is smooth in the sense of Gateaux if $R$ is
    a Fr\'echet ring.
\end{enumerate}
\end{Theorem}

\noindent In order to prove this theorem, we need to introduce some notation: 
we let $u \mapsto u_-$ be the linear projection on $\Psi^{-1}(A_t)$
with respect to the $\d-$components of $u \in
\widehat{\Psi}(A_t)$, and we set $u_+ = u-u_-\, .$

\begin{proof}
 Part (1) is already proved in \cite{M1,M3} (see also the later exposition 
 \cite{ER2013}), but we have to examine carefully
 the original proof in order to check that the map $U \mapsto S$ is smooth. 
 Once we have this first result,
 we finish the proof of the theorem by remarking that $Y = SU.$

    Let us follow the proof of \cite[Lemma 3]{M1}. We solve the equation 
    \begin{equation}\label{st1}
    (SU)_- = 0
    \end{equation}
    for a given $$U = \sum_{\nu \in \Z}u_\nu \d^\nu$$
    and unknown $$S = \sum_{\mu \in \Z-\N} s_\mu \d^\mu.$$
    Equation (\ref{st1}) turns into the infinite system
$$
\sum_{\nu \in \Z-\N}
\left(\sum_{i\in \N} 
\left(\begin{array}{c}\nu \\ i \end{array} \right)\d^i u_{\mu + i - \nu} \right) s_\nu 
               = -u_\mu \quad \hbox{ for } \mu \in \Z-\N\; ,
$$
with the convention 
$$
\left(\begin{array}{c}\nu \\ i \end{array} \right) = 
(-1)^{|\nu|} \left(\begin{array}{c}|\nu|\\ i \end{array} \right) \hbox{ if } \nu < 0\; .
$$
We solve this infinite system by induction on $\mu.$ For $\mu = -1,$ let us study the 
$t-$monomials.  
$$\left[\sum_{\nu \in \Z-\N}\left(\sum_{i\in \N} \left(\begin{array}{c}\nu \\ i \end{array} \right)\d^i u_{-1 + i - \nu} \right) s_\nu\right]_0 
= \left[\left(\sum_{i\in \N} \left(\begin{array}{c}-1 \\ i \end{array} \right)\d^i u_{ i } \right) s_{-1}\right]_0 
= -[s_{-1}]_0$$
thus $[s_{-1}]_0 = [u_{-1}]_{0}.$
Then, at the monomial of valuation 1, which is only the $t_1-$monomial,
    \begin{eqnarray*}
    \left[\sum_{\nu \in \Z-\N}\left(\sum_{i\in \N} \left(\begin{array}{c}\nu \\ i \end{array} \right)\d^i u_{-1 + i - \nu} \right) s_\nu\right]_{t_1} & = & \left[\left(\sum_{i\in \N} \left(\begin{array}{c}-1 \\ i \end{array} \right)\d^i u_{ i } \right)\right]_{t_1} \left[s_{-1}\right]_0 +\\&& \left[\left(\sum_{i\in \N} \left(\begin{array}{c}-1 \\ i \end{array} \right)\d^i u_{ i } \right)\right]_{0} \left[s_{-1}\right]_{t_1} \\
    & = & \left[\left(\sum_{i\in \N} \left(\begin{array}{c}-1 \\ i \end{array} \right)\d^i u_{ i } \right)\right]_{t_1} \left[u_{-1}\right]_0 - [s_{-1}]_{t_1}\end{eqnarray*}
we get 
$$[s_{-1}]_{t_1} = [u_{-1}]_{t_1} + \left[\left(\sum_{i\in \N} \left(\begin{array}{c}-1 \\ i \end{array} \right)\d^i u_{ i } \right)\right]_{t_1} \left[u_{-1}\right]_0\; , $$
which shows that the maps $u\mapsto [s_{-1}]_0$ and $u \mapsto [s_{-1}]_{t_1}$ are smooth. 
The same calculations lead to any $t-$monomial for $s_{-1},$ and they show that 
$u \mapsto s_{-1}$ is smooth.
Let us now assume that $s_{-1},...s_{-k}$ are determined and that they are smooth with 
respect to $U,$ and let us determine $s_{-k-1}.$ Then we write, for $\mu = -k-1,$
 $$\sum_{\nu \leq -k-1}\left(\sum_{i\in \N} \left(\begin{array}{c}\nu \\ i \end{array} \right)\d^i u_{-k-1 + i - \nu} \right) s_\nu = -u_{-k-1} - \sum_{\nu =-k}^{-1}\left(\sum_{i\in \N} \left(\begin{array}{c}\nu \\ i \end{array} \right)\d^i u_{-k-1 + i - \nu} \right) s_\nu $$
 noticing that the right side is fully determined. Then we carry out induction with 
 respect to the valuation on $t,$ for which we sketch the first two steps as we did 
 before:

 - for the $t=0$ term:
$$[s_{-k-1}]_0 = (-1)^{-k-1}\left[u_{-k-1} + \sum_{\nu =-k}^{-1}\left(\sum_{i\in \N} \left(\begin{array}{c}\nu \\ i \end{array} \right)\d^i u_{-k-1 + i - \nu} \right) s_\nu\right]_0,$$

 - for the $t_1$-monomial:
 \begin{eqnarray*}
[s_{-k-1}]_{t_1} &=& (-1)^{-k-1}\left(\left[u_{-k-1} + \sum_{\nu =-k}^{-1}\left(\sum_{i\in \N} \left(\begin{array}{c}\nu \\ i \end{array} \right)\d^i u_{-k-1 + i - \nu} \right) s_\nu\right]_{t_1} + \right. \\ && \left. \left[\left(\sum_{i\in \N} \left(\begin{array}{c}-1 \\ i \end{array} \right)\d^i u_{ i } \right)\right]_{t_1} \left[s_{-k-1}\right]_0\right) \end{eqnarray*}
Thus, we have a construction of $S$ made of smooth operations, so that the map 
$$U \mapsto S$$ is smooth. Hence the map $$U \mapsto Y = SU $$ is also smooth. 
This completes the proof.
\end{proof}

\subsection{Differentiable structures and regular subgroups of $G(\Psi(A_t))$}

Let us stress that we have established smoothness of the
operations on the groups considered by Mulase, but that we have
not proved the existence of an exponential map. This is
principally because it is difficult to show the existence of the
exponential map on $\widehat{\Psi}(A_t);$ also, because it is  
difficult to determine the tangent space
$T_1G(\widehat{\Psi}(A_t))$, and even more, to
differentiate an hypothetical adjoint action of
$G(\widehat{\Psi}(A_t))$ on $T_1G(\widehat{\Psi}(A_t))$. Most of
the difficulties come from the very general definition of 
$\widehat{\Psi}(A_t)$. This is why we construct a Fr\"olicher subalgebra
$\overline{\Psi}(A_t)\subset\widehat{\Psi}(A_t),$ motivated by the
next Lemma. We denote by $exp_t$ the exponential defined on formal 
series of elements of $\Psi(A).$ 
 
\begin{Lemma} \label{expt}
Let $(P_n)_{n \in \mathbb{N}^*} \in \Psi(A)^{\mathbb{N}^*}$ such that $order(P_n)\leq n.$  
Then
$$exp_t \left(\sum_{i \in \N^*} \tau_i P_i \right) \in\widehat{\Psi}(A)\; ,$$  
and the map
$(P_n)_{n \in \mathbb{N}^*} \mapsto exp_t(\sum_{i \in \N^*} \tau_i P_i)$ is smooth with
respect to the Fr\"olicher structures of $\Psi(A)^{\mathbb{N}^*}$ and 
$\widehat{\Psi}(A)$. Moreover, if we write
$$
exp_t \left(\sum_{i\in\N^*} \tau_i P_i \right)=\sum_{\alpha \in \Z} a_\alpha \d^\alpha\; ,
$$ 
then we have
$$\alpha \leq val_t a_\alpha$$
for all $\alpha \in \Z$.
\end{Lemma}
\begin{proof}
Each monomial of  $exp_t(\sum_{i \in \N^*} t_i P_i)$ is, up to a scalar 
multiplication, a finite sum of finite products of terms of the sequence $(P_n).$ Since 
multiplication is smooth on $\Psi(A),$ we get that the map  
$(P_n)_{n \in \mathbb{N}^*} \mapsto  exp_t(\sum_{i \in \N^*} t_i P_i)$ is smooth. 
Thus, we only need to check that 
$\forall \alpha \in \Z, \quad \alpha \leq val_t a_\alpha,$ which will ensure us that
$exp_t(\sum_{i \in \N^*} \tau_i P_i) \in \widehat{\Psi}(A).$ Setting $P_f(\N)$ the set of 
finite subsets of $\N,$ we get that for any $I \in P_f(\N)$ and any multi-index 
$(\alpha_i)_{i \in I},$ the $\left(\prod_{i \in I} t_i^{\alpha_i}\right)-$monomial 
$p_{(t_i),(\alpha_i)}$ is of ($\d-$maximal) order 
$$\alpha = \sum_{i \in I} \alpha_i\; ,$$ 
by property (\ref{val3}). Then, from the classical formulas that compute $a_\alpha,$ we 
obtain  
$$
val_t a_\alpha  \geq \sum_{I \in P_f(\N), \sum_I \alpha_i=\alpha} \alpha_i = \alpha\; ,
$$ 
using (\ref{val1}) and (\ref{val2}).
\end{proof}

\begin{Definition}\label{barpsi}
The regular space of formal pseudo-differential and differential operators of infinite 
order are, respectively, $\overline{\Psi}(A_t)$ and $\overline{\D}_{A_t}$, in which
\begin{equation} \label{aldef-reg}
\overline{\Psi}(A_t) = \left\{ \sum_{\alpha \in {\mathbb{Z}}}
    a_{\alpha}\,\d^{\alpha} \in \widehat{\Psi}(A_t) \; : 
    \, val_t(a_\alpha)\geq \alpha \right\} 
\end{equation}
    and
    \begin{equation}
    \overline{\D}_{A_t} = \left\{ P= \sum_{\alpha \in \mathbb{Z}}
    a_{\alpha}\,\d^{\alpha} : P \in\overline{\Psi}(A_t) \mbox{ and }
    a_\alpha=0 \mbox { for } \alpha <0 \right\} \; .
    \end{equation}
In addition, we define
    \begin{equation} \label{345-bar}
    G(\overline{\Psi}(A_t)) = \{ P \in \overline{\Psi}(A_t) :
    P\vert_{\tau=0}\in G_{A_t}  \}
    \end{equation}
    and
    \begin{equation} \label{346-bar}
    \overline{\D}_{A_t}^{\times} = \{ P \in \overline{\D}_{A_t} :
    P\vert_{\tau=0}=1 \} \; .
    \end{equation}
\end{Definition}

\smallskip

 Complementing Remark \ref{fring}, we note that if $R$ is a Fr\'echet ring,
 then $\overline{\Psi}(A_t)$ and $\overline{\D}_{A_t}$ are Fr\'echet
 algebras. In our general case we have:

\begin{Lemma} \label{Psi-bar}
    $\overline{\Psi}(A_t)$ is a Fr\"olicher algebra, and 
    $G(\overline{\Psi}(A_t))$ is a Fr\"olicher group.
\end{Lemma}
\begin{proof}
From Theorem \ref{Psihat} and Theorem \ref{Ghat}, we only have to prove that 
$\overline{\Psi}(A_t)$ and $G(\overline{\Psi}(A_t))$ are stable under algebraic operations. 
In order to do this, we only need to check some estimates:

    \begin{itemize}
        \item Let $(\sum_{\alpha \in \Z} a_\alpha \d^\alpha, \sum_{\alpha \in \Z} b_\alpha \d^\alpha) \in \overline{\Psi}(A_t)^2$. Then $$\sum_{\alpha \in \Z} a_\alpha \d^\alpha+ \sum_{\alpha \in \Z} b_\alpha \d^\alpha = \sum_{\alpha \in \Z} (a_\alpha +b_\alpha) \d^\alpha.$$
        Using property (\ref{val1}), $$val_t(a_\alpha + b_\alpha ) \geq \inf \left(val_t(a_\alpha) , val_t (b_\alpha )\right) \geq \alpha\; .$$
        \item Proceeding in the same way and with the same notations,
        $$\left(\sum_{\alpha \in \Z} a_\alpha \d^\alpha\right)\left( \sum_{\alpha \in \Z} b_\alpha \d^\alpha\right) = \sum_{(\alpha,\beta, k) \in \Z^2 \times \N} \frac{1}{k!} (\d^k a_\alpha)b_\beta  \d^{\alpha + \beta - k}$$
        gives:
        $$val_t \left((\d^k a_\alpha)b_\beta\right) \geq val_t\left( a_\alpha b_\beta\right)$$ and by property (\ref{val3}),
            $$val_t \left((\d^k a_\alpha)b_\beta\right) \geq \alpha + \beta \geq \alpha + \beta - k\; .$$
        \item   Mimicking the proof of Lemma \ref{Psi-Rt}, if $\sum_{\alpha \in \Z} a_\alpha \d^\alpha \in G(\overline{\Psi}(A_t)),$ then 
        $$ \left[(\sum_{\alpha \in \Z} a_\alpha \d^\alpha)^{-1}\right]_t = - \left[\sum_{\alpha \in \Z} a_\alpha \d^\alpha\right]^{-1}_0\sum_{t'+t''=t, t'\neq 0} \left[\sum_{\alpha \in \Z} a_\alpha \d^\alpha\right]_{t'}\left[(\sum_{\alpha \in \Z} a_\alpha \d^\alpha)^{-1}\right]_{t''}.$$ 
        Since $$\left[\sum_{\alpha \in \Z} a_\alpha \d^\alpha\right]_0 \in G_R,$$ we have
        $$val_t\left[\sum_{\alpha \in \Z} a_\alpha \d^\alpha\right]_0=0\;,$$ and hence, considering each sum of monomials of constant valuation $k,$
        \begin{eqnarray*}
        &&ord\left( \sum_{|t|=k}\left[\left(\sum_{\alpha \in \Z} a_\alpha \d^\alpha\right)^{-1}\right]_t \right)\\ & \leq &
         \sup_{|t|=k}ord\left(\left[\sum_{\alpha \in \Z} a_\alpha \d^\alpha\right]^{-1}_0\sum_{t'+t''=t, t'\neq 0}
         \left[\sum_{\alpha \in \Z} a_\alpha \d^\alpha\right]_{t'}\left[(\sum_{\alpha \in \Z} a_\alpha \d^\alpha)^{-1}\right]_{t''} \right) \\
        & \leq & 0+\sup_{|t|=k} \, ord\left(\sum_{t'+t''=t, t'\neq 0} \left[\sum_{\alpha \in \Z} a_\alpha \d^\alpha\right]_{t'}\left[(\sum_{\alpha \in \Z} a_\alpha \d^\alpha)^{-1}\right]_{t''} \right)\\
        & \leq & \sup_{|t|=k}\sup_{t'+t''=t, t'\neq 0} ord\left(\left[\sum_{\alpha \in \Z} a_\alpha \d^\alpha\right]_{t'}\left[(\sum_{\alpha \in \Z} a_\alpha \d^\alpha)^{-1}\right]_{t''} \right)\\
        & \leq & \sup_{|t|=k}\sup_{t'+t''=t, t'\neq 0} ord\left(\left[\sum_{\alpha \in \Z} a_\alpha \d^\alpha\right]_{t'}\right) + ord\left(\left[(\sum_{\alpha \in \Z} a_\alpha \d^\alpha)^{-1}\right]_{t''} \right)\\
            & \leq & \sup_{|t|=k}\sup_{t'+t''=t, t'\neq 0} val_t\left(\left[\sum_{\alpha \in \Z} a_\alpha \d^\alpha\right]_{t'}\right) + val_t\left(\left[(\sum_{\alpha \in \Z} a_\alpha \d^\alpha)^{-1}\right]_{t''} \right)\\
            & \leq &  \sup_{|t|=k}\sup_{t'+t''=t, t'\neq 0} |t'| + |t''| \\
            & \leq & k.
        \end{eqnarray*}
    \end{itemize}
\end{proof}

The following proposition gathers all important
properties of the ``barred" spaces introduced above, which are
necessary for our work. They are all consequences of Lemma
\ref{Psi-bar}. We leave the proof of this proposition as a toy
exercise for the reader. It can be proved either by direct
computations or by considering the intersection of the ``hatted"
sets appearing in Theorems \ref{Psihat} and \ref{Ghat} with
$\overline{\Psi}(A_t).$

\begin{Proposition} \label{Triv-bar}
    We have the following properties:
    \begin{enumerate}
        \item $\mathcal{I}_{A_t} \subset \overline{\Psi}(A_t)$.
        \item $G_{A_t} \subset G(\overline{\Psi}(A_t))$.
        \item $\overline{\D}_{A_t}^{\times}$ is a Fr\"olicher subgroup of $G(\overline{\Psi}(A_t)).$
        \item For any $U \in G(\overline{\Psi}(A_t))$ there exist unique $W \in
        G_{A_t}$ and $Y \in \overline{\D}_{A_t}^{\times}$ such that
        \[
        U = W^{-1}\,Y \; .
        \]
        In other words, there exists a unique global factorization of the 
        Fr\"olicher Lie group
        $G(\overline{\Psi}(A_t))$ as a group defined by matched pairs,
        \[
        G(\overline{\Psi}(A_t)) = G_{A_t} \, \overline{\D}_{A_t}^\times \; .
        \]
        \item The factorization $U \mapsto (W,Y)$ is smooth.
    \end{enumerate}
\end{Proposition}

We now turn to regular Lie groups. For this, and till the end of the paper, we assume that 
$A_t$ is a regular algebra, and that $A_t^*$ is a regular Lie group. This is equivalent to 
saying that $R$ is a regular vector space, and $R^*$ is a regular Lie group.

\begin{Lemma} \label{Psi-vreg}
    $\overline{\Psi}(A_t)$ is a regular vector space.
\end{Lemma}
\begin{proof}
Let $s\mapsto\sum_{n\in\Z} a_\alpha(s)\d^\alpha\in C^\infty([0;1],\overline{\Psi}(A_t)) .$
    \begin{itemize}
        \item If $val_t(a_\alpha(s))< \alpha$, then $a_\alpha(s) = 0$ 
    % if $|t|<\alpha$ then $a_\alpha(s) = 0$ 
    and hence $\int_0^sa_\alpha = 0.$
        \item If $val_t(a_\alpha(s))\geq \alpha$,
        %if $|t|\leq \alpha,$ 
        then $s \mapsto \int_0^sa_\alpha$ exists since $R$ is regular.
    \end{itemize}
    This integral does not change the inequalities on $val_t$ satisfied by 
    $a_\alpha(s)$, and so
$$s\mapsto \int_0^s \sum_{n \in \Z} a_\alpha(s)\d^\alpha \in C^\infty([0;1],
\overline{\Psi}(A_t))\; .$$ 
\end{proof}

\begin{Theorem} \label{Triv-d}
The group $\overline{\D}_{A_t}^{\times}$ is a regular Fr\"olicher Lie group, with regular 
Lie algebra  
$$
\overline{\mathfrak{d}}_{A_t}^{\times} = 
\{ P \in \widehat{\D}_{A_t} : P\vert_{\tau=0}=0 \}\; .
$$  
\end{Theorem}
\begin{proof} 
The vector space $\overline{\mathfrak{d}}_{A_t}^{\times}$ is by construction
a Fr\"olicher vector space, which is regular as can be seen by adapting the 
proof of Lemma \ref{Psi-vreg}.  From Proposition \ref{Triv-bar}, 
$\overline{\D}_{A_t}^{\times}$  is a Fr\"olicher group.
Let $$c(s) = a_0(s) + \sum_{\alpha \geq 0} a_\alpha(s)\d^\alpha$$ be a path 
in $\overline{\D}_{A_t}^{\times}$ starting at $1.$  From Equation (\ref{346-bar}) we have 
$$\d{a}_0(0) = 0\; .$$
Moreover, $val_t(a_\alpha) \geq \alpha$, and hence 
    $$val_t(d_s a_\alpha (0)) \geq \alpha\; .$$
Thus, 
$$\d_s{c}(0) \in \overline{\mathfrak{d}}_{A_t}^{\times}\; .$$
Conversely, let $w \in \overline{\mathfrak{d}}_{A_t}^{\times}$. Since $val_t(w) \geq 1,$
we have that the path 
$$c(s)= 1 + s w  \in\overline{\D}_{A_t}^{\times}\; .$$
Since $\overline{\mathfrak{d}}_{A_t}^{\times}$ is a Fr\"olicher algebra and 
$\overline{\D}_{A_t}^{\times}$ is a Fr\"olicher group, we conclude that the
Adjoint map 
\begin{eqnarray*}
    \overline{\D}_{A_t}^{\times} \times \overline{\mathfrak{d}}_{A_t}^{\times} 
    &\rightarrow& \overline{\mathfrak{d}}_{A_t}^{\times}\\
    (D,d) & \mapsto& DdD^{-1}
    \end{eqnarray*} 
is well-defined, with values in $\overline{\mathfrak{d}}_{A_t}^{\times},$ and it
differentiates to the standard Lie bracket 
\begin{eqnarray*}
\overline{\mathfrak{d}}_{A_t}^{\times} \times \overline{\mathfrak{d}}_{A_t}^{\times} 
&\rightarrow& \overline{\mathfrak{d}}_{A_t}^{\times}\\
    (d_1,d_2) & \mapsto& [d_1,d_2]=d_1d_2 - d_2d_1\; ,
    \end{eqnarray*} 
which shows that $\overline{\D}_{A}^{\times}$ is a Fr\"olicher Lie group with regular Lie 
algebra $\overline{\mathfrak{d}}_{A_t}^{\times}.$ We complete the proof by applying 
Theorem \ref{regulardeformation}, with $\A_n$ defined as the vector space of all 
$(t,\d)-$monomials of valuation equal to $n$, and of order $\leq n$.
\end{proof}

\noindent We now give the corresponding theorem for $G(\overline{\Psi}(A_t)).$

\begin{Theorem} \label{Psi-atreg}
 The group $G(\overline{\Psi}(A_t))$ is a regular Fr\"olicher Lie group, with regular 
 Fr\"olicher Lie algebra
 $$\mathfrak{g}(\overline{\Psi}(A_t)) = 
 \left\{ P \in \overline{\Psi}(A_t) : P\vert_{\tau=0} \in \Psi^{-1}(R)\right\}.$$
\end{Theorem}
\begin{proof}
    We have the following exact sequence:
$$
0 \rightarrow Ker((.)|_{\tau=0}) \rightarrow G(\overline{\Psi}(A_t)) \rightarrow 
G_R \rightarrow 0
$$
with smooth inclusion $G_R \rightarrow G(\overline{\Psi}(A_t))$ on the $\tau=0$ component.
So that by Theorem \ref{exactsequence} and Lemma \ref{Psi-1}, we only need to prove that
    $ Ker((.)|_{\tau=0})$ is a regular Lie group with Lie algebra
    $$\mathfrak{l} = \left\{P \in \overline{\Psi}(A_t) : P|_{\tau=0}=0 \right\}.$$
    We notice that
    $$Ker((.)|_{\tau=0}) = G\Psi_t(R)\cap \overline{\Psi}(A_t)$$
    and hence $Ker((.)|_{\tau=0})$ is a Fr\"olicher Lie group.
    Let us check the following:
\begin{itemize}
\item \underline{$\mathfrak{l}$ is regular as a vector space}
     It is proved in Lemma \ref{Psi-vreg} that any path in 
     $C^{\infty}([0,1],\mathfrak{l})$ integrates to a path on $\overline{\Psi}(A_t),$
     and it is easy to see that the $\tau=0$ component, which equals to 0, integrates 
     to $0.$
\item \underline{The exponential map $exp:C^{\infty}([0,1],\mathfrak{l}) \rightarrow
 C^{\infty}([0,1],Ker(P|_{\tau=0}))$ }
We can now apply Theorem \ref{regulardeformation} setting $\A_m$ as the vector space of 
$\tau-$monomials of valuation equal to $m,$ and which order is not greater than $m.$
\end{itemize}
\end{proof}

%{\color{red} JP, since it seems Lemma \ref{expt} is about an infinite number of 
%formal variables $\tau_i$, I think we can delete the remark below}

%\begin{rem}
%If we know a priori that $A_t$ is regular, then Theorem \ref{Psi-atreg}
%would imply Lemma \ref{expt}. This Lemma was stated by completeness, but it is not 
%employed in the proof of Theorem \ref{Psi-atreg}.
%\end{rem}

\section{Integration of the Kadomtsev-Petviashvili hierarchy} \label{S4}

The Kadomtsev-Petviashvili (KP) hierarchy reads
\begin{equation} \label{lolo}
\frac{d L}{d t_{k}} = \left[ (L^{k})_{+} , L \right]\; , \quad
\quad k \geq 1 \; ,
\end{equation}
with initial condition $L(0)=L_0  \in \d + \Psi^{-1} (R)$. The dependent
variable $L$ is chosen to be of the form
$$L = \d + \sum_{\alpha \leq -1 } u_\alpha \d^\alpha \in {\Psi}^1(A_t) \; .$$
A standard reference on (\ref{lolo}) is L.A. Dickey's treatise \cite{D}, see also 
\cite{KZ,M1,M3}. The following result gives a solution to the Cauchy problem
for the KP hierarchy (\ref{lolo}).

\begin{Theorem} \label{KPcentral}

~

\begin{enumerate}
\item
Consider the KP hierarchy 
\begin{equation} \label{lolo1}
\frac{d L}{d \tau_{k}} = \left[ (L^{k})_{+} , L \right]\; , \quad
\quad k \geq 1 \; ,
\end{equation}
with initial condition $L(0)=L_0$. Then,
\begin{enumerate}
\item There exists a pair $(S,Y) \in G_{A_t} \times \widehat{\D}(A_t)^\times$ 
        %$(S,Y) \in G_{A_t} \times \overline{\D}(A_t)^\times$  
    such that the unique solution to Equation $(\ref{lolo1})$ with $L|_{\tau=0}=L_0$ is
    \begin{eqnarray*}
        L(\tau_1,\tau_2,\cdots)=Y\,L_0\,Y^{-1} = S L_0 S^{-1} \; .
    \end{eqnarray*}
\item The pair $(S,Y)$ is uniquely determined by the smooth decomposition problem 
    $$exp\left(\sum_{k \in \N}\tau_k L_0^k\right) = S^{-1}Y$$ 
    and the solution $L$ depends smoothly on the initial condition $L_0$.
\end{enumerate}
\item Consider the KP hierarchy $(\ref{lolo})$
with initial condition $L(0)=L_0$. Set $\K^\infty=\cup_{n \in \N}\K^n$ and equip this set 
with the structure of a locally convex topological space given by the inductive limit. The 
algebra $A_t$ is to be considered as an algebra of power series in infinitely many 
variables, and each monomial of a given series in $A_t$ takes values in $\K^\infty$. 
\begin{enumerate}
\item There exists a pair $(S,Y) \in G_{A_t} \times 
\overline{\D}(A_t)^\times$ such that the unique solution to Equation $(\ref{lolo})$ with 
$L(0)=L_0$ is
\begin{eqnarray*}
        L(t_1,t_2,\cdots)=Y\,L_0\,Y^{-1} = S L_0 S^{-1} \; .
\end{eqnarray*}
\item The pair $(S,Y)$ is uniquely determined by the smooth decomposition problem 
    $$exp\left(\sum_{k \in \N} t_k L_0^k\right) = S^{-1}Y$$ 
posed in the regular Fr\"olicher Lie group $G(\overline{\Psi}(A_t))$.
\item The solution operator $L$ is smoothly dependent on the variable $t$ and on the 
    initial value $L_0.$ This means that the map 
$$ 
(L_0,s) \in (\d + \Psi^{-1}(R))\times \K^\infty \mapsto 
\sum_{n \in \N} \left( \sum_{|t|=n}[L(s)]_t\right) \in (\d + \Psi^{-1}(A_t))^\N 
$$ 
is smooth.
\end{enumerate}
\end{enumerate}
\end{Theorem}

The existence of an algebraic decomposition as in Parts 1(a), 1(b) 
of Theorem \ref{KPcentral} appears already in Mulase's seminal
paper \cite{M1}, and a formal solution to (\ref{lolo}) as in Part
(1) is in \cite{ER2013}. The richness of Theorem \ref{KPcentral}
steams from the fact that we pose the KP hierarchy in
the Fr\"olicher algebra $\Psi(A_t)$ (Theorem \ref{Psihat}) and
that then we {\em solve} the corresponding Cauchy problem using
analytically rigorous factorizations in Fr\"olicher Lie groups
({\em e.g.} Theorem \ref{Ghat}, Proposition \ref{Triv-bar}, Theorem
\ref{Triv-d} and Theorem \ref{Psi-atreg} for the infinite-dimensional 
regular Fr\"olicher Lie group $G(\overline{\Psi}(A_t))\,$). We
propose two proofs, one inspired on the formal arguments of
\cite{ER2013} and the other inspired on the work of \cite{Ma2013}
on a $q$-deformed KP hierarchy. 
As is plain from the discussions below, the new arguments used
herein are intrinsically interwoven with the ones used in
\cite{ER2013,Ma2013}.

\subsection{First proof}
We prove Part (2). Motivated by \cite{ER2013}, we model our proof after 
Reyman and Semenov-tian-Shansky's \cite{RS}, see also the later exposition \cite{P}. We 
set 
$$U = \exp \left(\sum_{k \in \N}t_kL_0^k\right)\; .$$
Then, $U \in G(\overline{\Psi}(A_t))$ by Lemma \ref{expt} applied to $P_n=L_0^n$ and
Theorem \ref{Psi-atreg}. 
We consider the smooth decomposition 
$$U \mapsto (S,Y) \in G_{A_t}\times \overline{\D}(A_t)^\times $$
following Proposition \ref{Triv-bar}, and we set $L=Y\,L_0\,Y^{-1}$.
We make two obvious observations:
\begin{enumerate} 
    \item
        $L^k = Y L_0^{\;k} Y^{-1}$.
    \item  $U\,L_0^{\;k} U^{-1}
    = L_0^{\;k}$, since $L_0$ commutes with 
    $U = \exp (\sum_k t_k\,L_0^{\; k}).$
    \end{enumerate}

     It follows that $L^k = Y L_0^{\;k} Y^{-1} = S S^{-1}
    Y L_0^{\;k} Y^{-1} S S^{-1} = S L_0^{\;k} S^{-1}$.

    We take $t_k$-derivative of $U$ for each $k \geq 1$. We obtain the 
    equation
    \[
    L_0^{\; k} U = - S^{-1}S_{t_k} S^{-1} Y + S^{-1} Y_{t_k}
    \]
    and so, using $U = S^{-1}\,Y$, we obtain the decomposition
    \[
    S L_0^{\; k} S^{-1} = - S_{t_k} S^{-1} + Y_{t_k} Y^{-1} \; .
    \]
    Since $S_{t_k} S^{-1} \in \mathcal{I}_A$ and $Y_{t_k} Y^{-1} \in
    \mathcal{D}_A$, we conclude that
    $$(L^k)_+ = Y_{t_k} Y^{-1} \; \; \mbox{ and } \; \; (L^k)_- = - S_{t_k}
    S^{-1}.$$ Now we take $t_k$-derivative of $L$:
    \begin{eqnarray*}
        \frac{d L}{d t_k} & = & Y_{t_k} L_0 Y^{-1} - Y L_0 Y^{-1} Y_{t_k}
        Y^{-1} \\
        & = & Y_{t_k} Y^{-1} Y L_0 Y^{-1} - Y L_0 Y^{-1} Y_{t_k}
        Y^{-1} \\
        & = & (L^k)_+\, L - L\, (L^k)_+ \\
        & = & [ (L^k)_+ , L ] \; .
    \end{eqnarray*}
   
    We check the initial condition: We have $L(0) = Y(0)
    L_0 Y(0)^{-1}$, but $Y(0) = 1$ since $Y \in \overline{\D}(A_t)^\times.$

    Smoothness with respect to $t$ is already proved by construction, and we 
    have established smoothness
    of the map $L_0 \mapsto Y$ at the beginning of the proof. Thus, the map
    $$L_0 \mapsto L(t)= Y(t) L_0 Y^{-1}(t)$$ is smooth in 
    $\overline{\Psi}(A_t)$ and 2(c) follows.

    In order to finish the proof, we need to prove that the solution $L$ 
    belongs to   
    $\d + \Psi^{-1}(A_t)$. For this, we consider above mentioned relation
    $$L^k = SL_0^k S^{-1},$$ which for $k=1,$ implies that 
    $$L = L_0 + S [L_0,S^{-1}].$$
 Then, $[L_0,S^{-1}]\in \Psi^{-1}(A_t),$ $S \in \Psi_0(A_t)$ and 
 by hypothesis $L_0 \in \d + \Psi^{-1}(R),$ so that 
 $$L \in \d + \Psi^{-1}(A_t).$$

    \begin{rem}
     The pseudo-differential operator $L$ we have constructed gives a full 
     solution to the KP hierarchy,
     in contradistinction to the result proven in \cite{ER2013} in which 
     (\ref{lolo}) is considered as
     an infinite system of equations for the coefficients of $L$ in {\em two} 
     independent variables,
     the ``time variable" $t_k$ and the ``space variable" modelled by the 
     derivation $\partial$.
     We can give another construction of the operators $S$, $Y$ and $L$ appearing
     above using induction and taking advantage of the fact that we are 
     considering our infinitely many time variables to take values in $\K^\infty$: 

        First, we fix $t_2 = t_3 = \cdots = 0,$ and we obtain a first 
        operator
        $U_1(t_1)= \exp(t_1L_0)=S_1^{-1}Y_1$ such that 
        $L_1 = Y_1 L_0 Y_1^{-1} = S_1 L_0 S_1^{-1}$
        is the unique solution of the $t_1-$equation of the KP-hierarchy
        $$\frac{d L_1}{dt_1} = \left[(L_{1})_+,L_1\right].$$
        Then, applying this first result to $R = R[[t_1]], $ and noticing 
        that $R[[t_1,t_2]]= R[[t_1]][[t_2]],$
        we set 
$$ 
U_2(t_1,t_2) = \exp(t_2L_1) =S_2^{-1}Y_2\in G\overline{\Psi}R[[t_1,t_2]] \; ,
$$
thereby obtaining a solution $$L_2(t_1,t_2) = Y_2 L_1 Y_2^{-1} = Y_2Y_1 L_0
Y_1^{-1}Y_2^{-1}$$ or equivalently
$$
L_2(t_1,t_2) = S_2 L_1 S_2^{-1} = S_2S_1 L_0 S_1^{-1}S_2^{-1}\; ,
$$ 
and so on. The total solution $L$ is built by a
inductive limit process which gives the operators $S$ and $Y$, and we can 
also easily check smoothness in $t$ of this full solution by remarking that 
$\sum_{|t|=n}[L(s)]_t$ is a finite sum of $t-$monomials.
\end{rem}

\subsection{Second proof}
In this second approach we model our proof
after \cite{Ma2013} and \cite{M3,M4}. As in \cite{Ma2013}, we apply a scaling to the KP hierarchy 
in order to work with fully regular Fr\"olicher Lie groups, and we make 
rigorous some of the formal arguments of \cite{M3,M4}.   

We let $L \in \d + \Psi^{-1}(A_t)$ and, as before, we consider the space 
$\K^\infty = \cup_{n \in \N}\K^n$ equipped with the structure of locally convex 
topological space given by the inductive limit. 
Since we have
    $[L^n,L]=0,$ we obtain from (\ref{lolo}) that 
\begin{equation}\label{lolo3}
\frac{d L}{dt_n} = [L^n_+,L] = -[L^n_-,L]\; , \quad \quad n \geq 1 \; .
\end{equation}
    Then, setting $dt = \sum_{k \in \N^*} d t_k$, we can define
    $$Z dt = \sum_{n \in \N^*} L^n_+ dt_n \quad \mbox{ and } \quad 
      Z^c dt = -\sum_{n \in \N^*} L^n_- dt_n\; .$$
The form $Z^c dt$ is a bona fide one-form, 
$$Z^c dt \in \Omega^1(\K^\infty,\Psi^{-1}({A_t}))\; ,$$ where $A_t$
    is understood as  set of smooth maps $\K^\infty \rightarrow R$ and  ${\Psi}(A_t)$ as a set of
    smooth maps  $\K^\infty \rightarrow \Psi(R).$
    Thus, in this context derivations with respect to
     $t_n$ appearing in (\ref{lolo3}) correspond to derivations 
    in a locally convex topological vector space. 

\noindent It follows from Equation (\ref{lolo3}) that the equation 
$$d Z^c - [Z^c,Z^c]=0$$ 
holds, and our constructions imply that this equation is a rigorous zero-curvature  
equation. 

\smallskip

In order to build a solution to (\ref{lolo3}) we proceed to deform the KP hierarchy and we 
work with deformed versions of $Z dt$ and $Z^c dt$. After \cite{Ma2013}, we use the 
scaling $t_n \mapsto q^n t_n$ to build the following operator  
    in ${\Psi}(A_t)[[q]]:$
    \begin{equation} \label{sca}
    \tilde{L}(t_1,t_2,...) = qL(qt_1,q^2t_2,...)\; .
    \end{equation}
We note that if we define $\delta : \K^\infty \rightarrow \K^\infty$ by 
$\delta(t_1,t_2,\cdots) = (q t_1,q^2 t_2,\cdots)$ then
\[
\frac{d \tilde{L}}{d t_n} = 
q \left.\frac{d (L \circ \delta)}{d t_n}\right|_{(t_1,t_2,\cdots)} = 
q^{n+1} \left.\frac{d L}{d t_n}\right|_{\delta(t_1,t_2,\cdots)} \; ,
\]
an easy observation which will be of use presently. We also consider the 1-forms
\begin{equation} \label{z+}
\tilde{Z}_{\tilde{L},+}(t_1,...) = \sum_{n = 1}^{+\infty} q^n (\tilde{L}^n)_+ dt_n
\end{equation}
and
\begin{equation} \label{z-}
\tilde{Z}_{\tilde{L},-}(t_1,...) = -\sum_{n = 1}^{+\infty}q^n (\tilde{L}^n)_- dt_n\; .
\end{equation}
By construction, $\tilde{Z}_{+}$ and $\tilde{Z}_{-}$ are smooth 1-forms in 
$\Omega^1(T,\Psi(A_t)[[q]])$. Now, instead of Equation (\ref{lolo3}),  
we write down an equation for $\tilde{L}$. Using (\ref{lolo3}) we obtain, after cancellation 
of the $q^{n+1}$-factors, the ``deformed KP hierarchy" 
\begin{equation} \label{lolo4}
\frac{d \tilde{L}}{dt_n}  =  [\tilde{L}^n_+\, ,\, \tilde{L}]
 =  - [\tilde{L}^n_-\, ,\, \tilde{L}]\; , \quad n \geq 1 \; .
\end{equation}  
These equations are equivalent to 
\begin{equation}\label{ast!!}
{d (\tilde{L})} = 
    [ \tilde{Z}_{\tilde{L},+} , \tilde{L} ] = [ \tilde{Z}_{\tilde{L},-} , \tilde{L} ] \; , 
\end{equation}
and we conclude that the corresponding zero curvature equations 
\begin{equation} \label{zcc+} 
d\tilde{Z}_{\tilde{L},+} + [\tilde{Z}_{\tilde{L}+},\tilde{Z}_{\tilde{L},+}] = 0
\end{equation}
    and
\begin{equation} \label{zcc-} 
d\tilde{Z}_{\tilde{L},-} - [\tilde{Z}_{\tilde{L},-},\tilde{Z}_{\tilde{L},-}] = 0
\end{equation}
also hold.

    We define the following algebras and groups, in which $val_q$ is 
    valuation of $q-$series, after \cite{Ma2013,Ma2015}:

    \begin{Definition}
        We set
        $$\Psi_q(R)=\left\{\sum_{\alpha \in \Z}a_{\alpha}\d^\alpha \in \Psi(R)[[q]] :
         val_q(a_\alpha)\geq \alpha\right\},$$
        $$G\Psi_q(R)=\left\{\sum_{\alpha \in \Z}a_{\alpha}\d^\alpha \in \Psi_q(R) : 
        a_0 = 1 + b_0, \quad val_q(b_0)\geq 1\right\},$$

\smallskip \smallskip

$$G_{R,q} = \left\{A \in  G\Psi_q(R) : A = 1 + B, \quad B\in \Psi^{-1}(R)[[q]] \right\},$$
$$\D_q(R) = \left\{\sum_{\alpha \in \Z}a_{\alpha}\d^\alpha \in \Psi_q(R) : a_\alpha = 0
 \hbox{ if } \alpha < 0 \hbox{ and }a_0 = 1 + b_0, \quad val_q(b_0)\geq 1 \right\},$$
    \end{Definition}

    The following Lemma is proved directly in \cite{Ma2013,Ma2015}, but now 
    we can give an alternative shorter proof by 
    using some of the results already stated in this paper.

    \begin{Lemma} \label{44}
$\Psi_q(R)$ is a Fr\"olicher algebra and a regular Lie algebra, and the groups 
$G\Psi_q(R),$ $G_{R,q}$ and $\D_q(R)$ are fully regular Fr\"olicher 
Lie groups with Lie algebras given respectively by:
$$\mathfrak{g}\Psi_q(R)=
\left\{\sum_{\alpha\in\Z}a_{\alpha}\d^\alpha\in\Psi_q(R) :  val_q(a_0)\geq 1\right\}\; ,$$
        $$\mathfrak{g}_{R,q}= \Psi^{-1}(R)[[q]]\; , $$
        and
$$\mathfrak{d}_q(R) = 
\left\{\sum_{\alpha \in \Z}a_{\alpha}\d^\alpha \in D_q(R) : 
a_0 = 0 \hbox{ if } \alpha < 0 \quad val_q(a_0)\geq 1\right\} \; .$$
\end{Lemma}
\begin{proof}
        We simply remark that by substituting $q^n$ for $t_n$, we obtain maps
        $$ \overline{\Psi}(A_t) \rightarrow \Psi_q(R)\; ,$$
        $$ G(\overline{\Psi}(A_t)) \rightarrow G\Psi_q(R)\; ,$$
        $$ G_{A_t} \rightarrow G_{R,q}\; ,$$   
        and $$ \overline{\D}(A_t) \rightarrow \D_q(R)\; .$$
        By straightforward computation on the coefficients of the series,
        the first map $ \overline{\Psi}(A_t) \rightarrow \Psi_q(R)$ is a
        (smooth) morphism of Fr\"olicher algebras which commutes with the  valuations
        $val_t$ and $val_q,$ and with integration along
        paths. Thus, identifying as formal variables {\em e.g.,} 
        $t_1$ with $q,$ we have slice maps 
$$ \overline{\Psi}(A_t) \leftarrow \Psi_q(R)\; ,$$
$$ G(\overline{\Psi}(A_t)) \leftarrow G\Psi_q(R)\; ,$$
$$ G_{A_t} \leftarrow G_{R,q}\; .$$ Pulling back the Fr\"olicher structure of 
$\overline{\Psi}(A_t)$ we see that  $\Psi_q(R)$ is a regular Fr\"olicher algebra. 
       The rest of the Lemma follows easily using similar arguments.
    \end{proof}

   Now we adapt some arguments of \cite{M3,M4} to our Fr\"olicher context. We remark
   that the operator $\tilde{L}$ we have considered so far is not, at this step
   of the proof, a solution to (\ref{lolo4}), but simply a $q$-deformation of an operator
   $L \in \d + \Psi^{-1}(A_t)$. Our construction of a true solution to the KP
   hierarchy satisfying a given initial condition $L_0$ is based on the observation that
   the operator $q \partial$ {\em is} a (stationary) solution to our deformed KP hierarchy 
   (\ref{lolo4}). 
    
   We consider the one-forms $Z_{q\d,+}$ and $Z_{q\d,-}$ defined as in (\ref{z+}) and
   (\ref{z-}) with $q\d$ instead of $\tilde{L}$. Then the equation
\begin{equation} \label{zcc+d} 
d\tilde{Z}_{q\d,+} + [\tilde{Z}_{q\d,+},\tilde{Z}_{q\d,+}] = 0
\end{equation}
obviously holds, while 
\begin{equation} \label{zcc-d} 
d\tilde{Z}_{q\d,-} - [\tilde{Z}_{q\d,-},\tilde{Z}_{q\d,-}] = 0
\end{equation}
is a trivial identity. Equation (\ref{zcc+d}) allows us to apply the Ambrose-Singer
theorem \ref{Hslice} on the trivial principal bundle $\K^\infty \times 
G\Psi_q(R)$. More precisely, we choose an initial condition 
$L_0 \in \d + \Psi^{-1}(R)$ and we 
choose $S_0 \in G_{R,q}$ such that  $q L_0 = S_0 (q \d) S_0^{-1}.$
Theorem \ref{Hslice} implies that there exists a unique section 
\[
\tilde{U}(t) : \K^\infty \rightarrow \K^\infty \times G\Psi_q(R)
\]
such that
\begin{equation}\label{flat}
   d\,\tilde{U}\cdot \tilde{U}^{-1} = \tilde{Z}_{q\d,+} 
\end{equation}
with initial condition $\tilde{U}(0) = S_0^{-1} \cdot Y_0$, in which $Y_0=1$. This smooth
section $\tilde{U}$ is a ``$q$-deformation" of the dressing operator present in our first
proof, and it replaces Mulase's formal operator $U$ appearing in \cite[Theorem 1.4]{M3}.

\smallskip

Now we use the slice maps built in
the proof of Lemma \ref{44} and we obtain $\overline{U}(t) \in G(\overline{\Psi}(A_t))$
satisfying $\overline{U}(0)=\overline{S_0}^{-1}$, in which $\overline{S_0}$ is the lift
of $S_0$. We use our factorization result (Proposition 3.19) to obtain unique 
$\overline{S}(t) \in G_{A_t}$ and $\overline{Y}(t) \in \overline{\D}_{A_t}^\times$ 
satisfying the equation
$$\overline{U}(t) = \overline{S}(t)^{-1} \cdot \overline{Y}(t) \; ,$$
and the initial conditions $\overline{S}(0) = \overline{S}_0$, and $\overline{Y}(0) = 1$. 
We return to $G\Psi_q(R)$ (using again the maps appearing in the proof of Lemma 4.4), and 
we obtain the unique factorization
\begin{equation} \label{uuu}
\tilde{U}(t) = \tilde{S}(t)^{-1} \cdot \tilde{Y}(t)
\end{equation}
with $\tilde{S}(t) \in G_{R,q}$, $\tilde{Y}(t) \in \D_q(R)$, $\tilde{S}(0)=S_0$, and 
$\tilde{Y}(0)=1$. 
In order to finish this second proof, we set  
\begin{equation} \label{newl}
\tilde{L} =  \tilde{S} (q \d) \tilde{S}^{-1} \in C^\infty(T,\Psi_q(R)) 
\end{equation}
(so that $\tilde{L}$ satisfies the $q$-deformed initial condition $\tilde{L}(0) = q L_0$) 
and we also {\em define} the smooth one-forms  
\begin{eqnarray}
\tilde{Z}_+ & = & \tilde{S} \tilde{Z}_{q\d,+} \tilde{S}^{-1} + 
                   d \tilde{S}\,.\,\tilde{S}^{-1}\; , \label{defz+} \\
\tilde{Z}_- & = & d \tilde{S}\,.\,\tilde{S}^{-1}\; .   \label{defz-}
\end{eqnarray}
The one-form $\tilde{Z}_{q\d,+}$ satisfies (\ref{zcc+d}) and $\tilde{Z}_+$ is obtained 
from $\tilde{Z}_{q\d,+}$ via a gauge transformation, so that $\tilde{Z}_+$ satisfies the 
zero curvature equation (\ref{zcc+}). Also, $\tilde{Z}_-$ satisfies (\ref{zcc-}), since  
it is ``pure gauge". 

\smallskip
 
By construction $\tilde{Z}_-$ takes values in the Lie algebra of $G_{R,q}$, and we claim
that $\tilde{Z}_+$ takes values in the Lie algebra of $\D_q(R)$. Indeed, we compute using
Equation (\ref{uuu}):
\[
d \tilde{U} =  d ( \tilde{S}^{-1} \cdot \tilde{Y} )
=  - \tilde{S}^{-1} d \tilde{S}\, \tilde{S}^{-1} \tilde{Y} + \tilde{S}^{-1}d \tilde{Y}\; ; 
\]
multiplying by $\tilde{U}^{-1}$ on the right and using (\ref{flat}) we obtain
\[
\tilde{Z}_{q\d,+} = 
- \tilde{S}^{-1}\, d \tilde{S} + \tilde{S}^{-1} d \tilde{Y}\, \tilde{U}^{-1} \; .
\]
Then
\[
\tilde{S} \tilde{Z}_{q\d,+} \tilde{S}^{-1} = 
- d \tilde{S} \tilde{S}^{-1}  + d \tilde{Y}\,\tilde{U}^{-1} \tilde{S}^{-1}  \; .
\]
Adding $d \tilde{S}\, \tilde{S}^{-1}$ to both sides of this equality and using
(\ref{defz+}), we obtain that $\tilde{Z}_+$ can be written in terms of the slice
$\tilde{Y} \in C^\infty (\K^\infty , \D_q(R))$ as
$$ \tilde{Z}_+ = d \tilde{Y} \cdot \tilde{Y}^{-1} \; ,$$ 
which proves our claim. 

\smallskip

We are ready to check that $\tilde{L}$, as defined in (\ref{newl}), satisfies the 
$q$-deformed KP hierarchy. Indeed:
\begin{eqnarray*}
\tilde{Z}_+ - \tilde{Z}_- & = & \tilde{S} \tilde{Z}_{q\d,+} \tilde{S}^{-1} 
   \; = \; \sum_{n=1}^\infty q^n (\tilde{S} q \d \tilde{S}^{-1})^n d t_n 
   \; = \; \sum_{n=1}^\infty q^n \tilde{L}^n  d t_n  \\
   & = & \sum_{n=1}^\infty q^n \left[ (\tilde{L}^n)_+ - (-\tilde{L}^n)_- \right] d t_n\; ,
\end{eqnarray*}
so that
$$
\tilde{Z}_{+} = \sum_{n = 1}^{+\infty} q^n (\tilde{L}^n)_+ dt_n \; 
\mbox{ and } 
\; \tilde{Z}_{-} = -\sum_{n = 1}^{+\infty} q^n (\tilde{L}^n)_- dt_n\; .
$$
Now, the operator $q \d$ is covariantly constant with respect to $\tilde{Z}_{q\d,+}$, this is,
$$
d (q \d)  = [\tilde{Z}_{q\d,+} , q \d] \; ,
$$
and therefore applying a gauge transformation with gauge $\tilde{S}$ we obtain
$$
d (\tilde{S} (q \d) \tilde{S}^{-1})  = [\tilde{S} \tilde{Z}_{q\d,+} \tilde{S}^{-1} + 
                   d \tilde{S}\,.\,\tilde{S}^{-1} , \tilde{S} (q \d)\tilde{S}^{-1}] \; ,
$$
this is,
$$
d (\tilde{L})  =  [ \tilde{Z}_{+} , \tilde{L}] =  \left[ \sum_{n = 1}^{+\infty} q^n (\tilde{L}^n)_+ dt_n \, , \tilde{L} \right] 
  =  [ \tilde{Z}_{\tilde{L},+} , \tilde{L}] \; .
$$
This equation says that the operator $\tilde{L}$ defined in (\ref{newl}) satisfies the first  
equation appearing in (\ref{ast!!}), which we know is equivalent to the deformed KP hierarchy 
(\ref{lolo4}). In an analogous fashion, using that $q\d$ is covariantly constant with respect to 
$\tilde{Z}_{q\d,-} = 0$, we can check that
$$
d (\tilde{L}) =  \left[ - \sum_{n = 1}^{+\infty} q^n (\tilde{L}^n)_- dt_n \, , \tilde{L} \right] 
  =  [ \tilde{Z}_{\tilde{L},-} , \tilde{L}] \; ,
$$
this is, $\tilde{L}$ given by (\ref{newl}) also satisfies the second equation appearing in 
(\ref{ast!!}).

\smallskip

It remains to check that our solution is unique. The only arbitrary choice made in this proof 
is the choice of $S_0 \in G_{R,q}$ such that  $q L_0 = S_0 (q \d) S_0^{-1}.$ The fact that 
$\tilde{L}$ is independent on $S_0$ follows from the following argument motivated by \cite{M4}: 

We write 
$$
\tilde{L} = \tilde{S} ( q \d ) \tilde{S}^{-1} = \tilde{Y} \tilde{U}^{-1} (q \d) \tilde{U} 
\tilde{Y}^{-1}
$$
and we claim that $\tilde{U}^{-1} (q \d) \tilde{U}$ is constant. Indeed, the 
equation 
$$\tilde{Z}_{q\d,+}= \sum_{n=1}^{+\infty} q^n (q\d)^n d t_n = 
d \tilde{U} \cdot \tilde{U}^{-1} = \sum_{n=1}^{+\infty} \frac{\d \tilde{U}}{\d t_n} 
\tilde{U}^{-1} d t_n$$ implies  
$$\frac{\d \tilde{U}}{\d t_n} \tilde{U}^{-1} = q^{2n} \d^n$$
and so
$$\frac{\d}{\d t_n} \left[ \tilde{U}^{-1} (q \d) \tilde{U} \right] = 
- \tilde{U}^{-1} \frac{\d \tilde{U}}{\d t_n}\tilde{U}^{-1} (q \d) \tilde{U} + 
\tilde{U}^{-1} q \d \frac{\d \tilde{U}}{\d t_n} = 0 \; .
$$ 
It follows that 
$$
\tilde{L} = \tilde{Y} \tilde{U}^{-1}(0)\, (q \d)\, \tilde{U}(0) \tilde{Y}^{-1} = \tilde{Y} S_0 
\,(q \d)\, S_0^{-1} \tilde{Y}^{-1} \; ,
$$
that is, our solution $\tilde{L}$ given by (\ref{newl}) can be also written as
$$\tilde L = \tilde Y (q L_0) \tilde Y^{-1} \; ,$$
and this expression shows that $\tilde L$ does not depends on $S_0$.

\smallskip

Finally, we consider the infinite jets
$$S_q = j^\infty(\tilde S)\; , $$
$$ Y_q =  j^\infty(\tilde Y) \; , $$
and $$ L_q =  j^\infty(\tilde L) \; .$$

\noindent By classical properties of infinite jets\footnote{Jets are considered in 
\cite{W}, in the classical paper \cite{BR} and in the more recent review \cite{R-ams}; we note 
that in \cite{W} the author introduces jets within the category $DS$ of {\em differential spaces}, but it 
is known that Fr\"olicher spaces form a subcategory of $DS$, see \cite{Ma2013}. 
Moreover, as explained in Section 2.5 (see also \cite{Ma2013} and \cite{F}) the category of 
Fr\"olicher spaces is cartesian closed, complete and cocomplete, and so infinite jets in Fr\"olicher spaces can 
be defined by adapting the constructions of \cite{BR,R-ams}.}, all relations 
above involving relations on $\tilde S,$ $\tilde Y$ and $\tilde L$
remain valid for $S_q,  Y_q$ and $ L_q,$ and the $t-$monomials of
the infinite jet series are polynomials in $q.$ Then, they are
$q-$smooth, where $q\in \K$ is no longer a formal variable.
Setting

$$ (S,Y,L) = \lim_{q \rightarrow 1} (S_q,Y_q,L_q)\; ,$$
we recover Theorem \ref{KPcentral}.

\section{Perturbed solutions and Hamiltonian approach} \label{S5}

\subsection{Hamiltonian formulation of KP}

We consider a regular Fr\"olicher Lie algebra of formal
pseudo-differential operators $\Psi(A)$ in which $A$ is an
arbitrary Fr\"olicher algebra equipped with a smooth derivation $\d$. First of all, we
recall that if $P\; = \sum_{- \infty < \nu \leq N} a_{\nu} \,
\partial^{\nu} \in \Psi(A)$, the {\em residue} of $P$ (see \cite{Ad}) is
\[
res (P) = a_{-1} \; .
\]
The most important property of $res$ (see Adler, \cite{Ad}; proofs
also appear in \cite{D} and \cite[Chapter 5]{O}) is the following:

\begin{Lemma}  \label{traceprop}
    If $P,Q \in \Psi(A)$, then $res([P,Q]) = \d(f)$ for some $f \in A$.
\end{Lemma}

Let us suppose that there exists a ${\mathbb K}$-linear function
$I : A \rightarrow B$, in which $B$ is an unitary Fr\"olicher
algebra, such that
\begin{equation}  \label{condition3}
I (\partial u) = 0 \mbox{ for all } u \in A \: .
\end{equation}
We define the  {\em trace form} on $\Psi (A)$ by
\begin{equation}
\mbox{Trace}(P) = I ( {\rm res}(P) ) \; \mbox{ for all } \; P
\in \Psi(A) \;  . \label{trace}
\end{equation}
{\em Hereafter we assume}, generalizing the standard case in which
$A = C^\infty(S^1,\mathbb{K})$ and the map $I$ is definite
integration, that the trace form (\ref{trace}) satisfies the
non-degeneracy condition
\begin{equation}\label{product}
\begin{array}{l}
\mbox{ The pairing } \left<\; ,  \, \right> : \Psi(A) \times
\Psi(A) \rightarrow B \mbox{ given
    by }  \left< P , Q \right> = \mbox{Trace}(PQ) \\
\mbox{ is non-degenerate, that is, } \left< P , Q \right> =0
\mbox{ for all } Q \in \Psi(A) \mbox{ implies } P = 0\; .
\end{array}
\end{equation}

Condition (\ref{condition3}) implies that we can ``integrate by parts'', 
that is, the identity $I( \d(u) \cdot v) = - I (u \cdot \d(v))$ holds for 
all $u,v \in A$.
This fact is crucial for the proof of the following standard lemma:
\begin{Lemma}
    The $B$-valued pairing $ \left<\; ,  \, \right>$ defined in $(\ref{product})$
    is $\mathbb{K}$-bilinear, symmetric, non-degenerate, and it satisfies
    \begin{equation} \label{47*}
    \left< [P,Q] , S \right> \; = \; \left< [S, P] , Q \right>
    \end{equation}
    for all $P,Q,S \in \Psi(A)$.
\end{Lemma}

Now we consider the Fr\"olicher Lie algebra $\Psi(A)$ and we
define the {\em
    regular dual space} $$\Psi(A)' = \{ \mu \in L(\Psi(A),B) : \mu =
\left< P , \cdot \right> \mbox{ for some } P \in \Psi(A) \}\; .$$
These definitions allow us to adapt standard results of
Hamiltonian mechanics (see for instance \cite{MR} or, the recent
summary \cite{ER2013} and references therein) to our
infinite-dimensional context as follows:

\smallskip

Inspired by \cite{GD1981,KO1990}, we let $f : \Psi(A)' \rightarrow
B$ be a polynomial function of the type

\begin{equation}
\label{Pol}f(\mu) = \sum_{k=0}^n a_k Trace(P^k)
\end{equation}
with $\mu = \left< P,. \right>.$ The {\em functional derivative of
f at} $\mu \in \Psi(A)'$ is the unique element $\delta f/\delta
\mu$ of $\Psi(A)$ determined by
\begin{equation}
\left< \nu \; \left| \; \frac{\delta \, f}{\delta \, \mu} \right> \right.
\; = d\, f_\mu ( Q )  \label{gradiente}
\end{equation}
for all $\nu = \left<Q,.\right> \in \Psi(A)'.$ If $f
=Trace\left(\sum_{k=0}^n a_k P^k\right),$ direct computations show
that $\frac{\delta f}{\delta \mu}$ is given by (algebraic)
derivation of the polynomials (\ref{Pol}) in agreement with \cite{D}.

\begin{Definition} \label{lpb}
    The $B$-valued Lie-Poisson bracket on the regular dual space
    $\Psi(A)'$ is defined as follows: for all polynomial functions $F, G :
    \Psi(A)' \rightarrow B$ of the type (\ref{Pol}) and $\mu \in \Psi(A)'$,
    \begin{equation}
    \{F \, , \, G \}(\mu) = \left< \mu \, \; \left|\; \left[
    \frac{\delta \, F}{\delta \, \mu} \, ,
    \, \frac{\delta \, G}{\delta \, \mu} \right] \right> \right. \; . \label{lie-poisson}
    \end{equation}
\end{Definition}

Now we fix a polynomial function $H : \Psi(A)' \rightarrow B$ with $H$ as
in (\ref{Pol}). Then, Equation (\ref{lie-poisson}) determines a
derivation $X_H$ on polynomial functions $F : \Psi(A)' \rightarrow
B$ via
\begin{equation} \label{vf}
X_H (\mu) \cdot F = \{ F , H \}(\mu)  = \left< \mu \; \left| \;
\left[ \frac{\delta \, F}{\delta \, \mu} \, ,
\, \frac{\delta \, H}{\delta \, \mu} \right] \right> \right.
\end{equation}
for all $\mu \in \Psi(A)'$.
\begin{Lemma} \label{le1}
    Let $H : \Psi(A)' \rightarrow B$ be a smooth function on
    $\Psi(A)'$. Then, with the previous notations, Hamilton's equations
    \begin{equation} \label{he}
    \frac{d }{dt} (F\circ \mu) = X_H(\mu) \cdot F
    \end{equation}
    for $\mu(t) = \left< P(t) , \cdot \right> \in \Psi(A)'$
    can be written as equations
    on $\Psi(A)$ as follows:
    \begin{equation} \label{lax1}
    \frac{d \, P}{d t} = \left[ \frac{\delta \, H}{\delta \, \mu} \, , \, P \right] \; .
    \end{equation}
\end{Lemma}
\begin{proof}
Let $P \in C^\infty(\R, \Psi(A))$ and let 
$\mu = \left<P,.\right> \in C^\infty(\R, \Psi(A)').$
We set up the differential equation
    \[
    \frac{d }{dt} (F\circ \mu) = X_H(\mu) \cdot F
    \]
for $\mu(t) = \left< P(t) , \cdot \right> \in \Psi(A)'$ and
$F$ as above. If we set $F(\mu) =Trace\left(\sum_{k=0}^n a_k P^k\right),$ we can compute 
the left and the right side of (\ref{he}) separately. Using (\ref{vf}) and (\ref{47*}) we 
obtain
\begin{eqnarray*}
\frac{d }{dt} (F\circ \mu) &=& \frac{d }{dt}Trace\left(\sum_{k=0}^n a_k P^k\right)\\ 
  &=& \sum_{k=0}^n a_kTrace \left(\frac{d P^k}{dt}\right)\\
  &=& \sum_{k=0}^n a_kTrace \left(kP^{k-1}\frac{d P}{dt}\right)\\ 
  &=& \left<  \frac{d P}{dt} ,  \frac{\delta \, F}{\delta \, \mu} \right> 
\end{eqnarray*}
    and
\begin{eqnarray*}
X_H(\mu) \cdot F &=& 
\left< P,\left[ \frac{\delta \, F}{\delta \, \mu} \, , \, 
                                   \frac{\delta \, H}{\delta \, \mu} \right] \right>\\
&=& \left<\left[ P,\frac{\delta \, H}{\delta \, \mu}\right],
        \frac{\delta \, F}{\delta \, \mu}   \right> \; ,
\end{eqnarray*}
    and the result follows.
\end{proof}
This proof suggests that we can identify $X_H(\mu)$ with the germ
$\frac{dP}{dt}.$ The equations for the integral curves of $X_H$
are Hamilton's equations on $\Psi(A)'$ corresponding to the
Hamiltonian function $H$.

Now we work in the context of $r$-matrices, see \cite{STS}. The
decomposition $\Psi(A) = \Psi^{-1}(A) \oplus \mathcal{D}_A$, in
which $\mathcal{D}_A =  \Psi(A) - \Psi^{-1}(A)$, allows us to
consider a new Lie bracket on the regular dual space $\Psi(A)'$
given by
\begin{equation}
[ P , Q ]_{0} = [ P_{+} , Q_{+} ] - [ P_{-} , Q_{-} ] \; , \label{r-bracket}
\end{equation}
in which $P_{\pm} = \pi_{\pm}(P)$, $Q_{\pm} = \pi_{\pm}(Q)$, and
$\pi_\pm$ are the projection maps from $\Psi(A)$ onto
$\mathcal{D}_A$ and $\mathcal{ I}_A$ respectively. This bracket
determines a new Poisson structure $\{ \, , \, \}_{0}$ on
$\Psi(A)'$, simply by replacing the original Lie product for $[\;
,\;]_0$ in (\ref{lie-poisson}). Using again the non-degenerate
pairing (\ref{product}) we obtain the following version of Lemma
\ref{le1}:

\begin{Lemma}   \label{adler2}
Let $H : \Psi(A)' \rightarrow B$ be a smooth function on $\Psi(A)'$ such that
    \begin{equation}
    \left< \, \mu \; \left| \; \left[ \frac{\delta H}{\delta \mu} \, ,
    \, \cdot \, \right] \right> \right. \, = 0   \quad \quad \mbox{
        for all } \mu \in \Psi(A)' \; .
    \label{ad1}
    \end{equation}
    Then, as equations on $\Psi(A)$, the Hamiltonian equations of motion
    with respect to the $\{ \, , \, \}_{0}$ Poisson structure of
    $\Psi(A)'$ are
    \begin{equation}
    \frac{d\,P}{d\,t} = 
\left[ \left( \frac{\delta H}{\delta \mu} \right)_{+} \, , \, P \right] \; .    
\label{lax2}
    \end{equation}
\end{Lemma}

\smallskip

We now use some specific functions $H$. Let us recall the
following results (see for example \cite{D} or the more recent
review \cite{ER2013}):

\begin{Proposition} \label{casimir}
    We define the functions $\displaystyle H_{k}(L) = Trace \left(
    (L^{k}) \right)$,
    $k=1,2,3,\cdots ,$ for $L \in \Psi(A)$. Then, $\displaystyle
    \frac{\delta H_{k}}{\delta L} = k L^{k-1}$. In particular, the functions $H_k$ satisfy 
    $(\ref{ad1})$.
\end{Proposition}

It follows that we can apply Lemma \ref{adler2}. It yields:

\begin{Proposition} \label{kpp}
    Let us equip the Lie algebra $\Psi(A)$ with the non-degenerate pairing
    $(\ref{product})$. Write $\Psi(A) = \Psi^{-1}(A) \oplus \mathcal{D}_A$
    and consider the Hamiltonian functions
    \begin{equation}
    \mathcal{H}_{k}(\mu) = \frac{1}{k} \, Trace \left( (L^{k+1}) \right) 
    \label{kpham}
    \end{equation}
    for $\mu = \left< L,.\right>$. The corresponding Hamiltonian equations of
    motion with respect to the $\{ \, , \, \}_{0}$ Poisson structure
    of $\Psi(A)'$ are
    \begin{equation}
    \frac{d L}{d t_{k}} = \left[ (L^{k})_{+} , L \right] \; .  \label{kp}
    \end{equation}
\end{Proposition}

Equations (\ref{kp}) are the Kadomtsev--Petviashvili hierarchy on
the space $\Psi(A)$.

\subsection{Perturbed solutions}

From now onwards, we set $R = C^\infty(S^1,\K)$ where $\K$ is $\R$
or $\mathbb{C}$ and we consider the algebra 
$$
R_z = \left\{\sum_{n \in \N} z^n a_n : (a_n)_\N\in
R^\N\right\}
$$
in the whole previous picture. The algebra $R_z$ is a Fr\'echet
algebra with set of units given by
$$
R_z^* =  \left\{\sum_{n \in \N} z^n a_n \in R_z : a_0 \in R^*\right\} \; ,
$$
as it can be seen using standard techniques on formal series. 
According to Remark \ref{comp}, smoothness in the Fr\"olicher
sense is equivalent to smoothness in the sense of
Gateaux derivatives in the category of Fr\'echet spaces. As a
consequence, The operator $\d$ on $R$ extends (smoothly) componentwise,
to $R_z$. 

We wish to develop KP in this setting, but before proceeding, let us
give our motivation.

It is well-known that there exists a theory of pseudo-differential
operators with symbols of limited smoothness, see e.g.
\cite{BR1984,Marsch1988,Ma2015}, but in this case the algebra of formal 
symbols cannot be constructed as in the standard case considered 
{\em e.g.} in \cite{Om} and references therein, principally because of the 
lack of continuity of the multiplication of functions in Sobolev
classes for low Sobolev orders. However, for functions $f \in H^s_0,$ there exists a sequence
$(f_n)_{n \in \N} \in (C^\infty_c)^\N$ such that $\lim f_n = f$ in $H^s.$
Setting $u_0 = f_0 $ and $ u_n = f_n - f_{n-1}$ for $n>0, $ we obtain
$$
f = \lim_{z \rightarrow 1} \sum_{n \in N^*} z^n u_n \; , 
$$
so that we can naturally investigate the KP hierarchy on spaces of
formal pseudo-differential operators with coefficients with low
regularity by using $R_z,$ as already sketched in \cite{Ma2015}
with other notations.

\smallskip

In order to ensure that we can apply our previous results in this
setting, we have to check:

\begin{enumerate}
    \item $R_z$ is a commutative $\K-$algebra with unit $1$;
    \item  $\d$ is a derivation on $R_z$;
    \item $R_z$ is a {Fr\"olicher algebra};
    \item $\d$ is smooth;
    \item $R^*_z$ is a Fr\"olicher Lie group with Lie algebra $\mathfrak{g}_R R_z$;   
    \item $R_z$ is regular as a Fr\"olicher vector space.
\end{enumerate}

The proofs of (1)-(4) and (6) are straightforward. The proof of
(5) is an application of Theorems \ref{regulardeformation} and
\ref{exactsequence} to the short exact sequence

$$
1 \longrightarrow \left\{\sum_{n \in \N}z^na_n \in R_q : a_0 = 1
\right\} \longrightarrow R_z^* \longrightarrow R^* \longrightarrow
1\; ,
$$
as we already did for other groups of series. 

Let us consider the Hamiltonian formulation of KP in this context.
We define the function $I : R_z \rightarrow \mathbb{R}_z =
\R[[z]]$ as
\[
I \left( \sum_{n \in \mathbb{N}} z^n f_n \right) = \sum_{n \in \mathbb{N}}
z^n \int_{S^1} f_n \; .
\]
The function $I$ clearly satisfies (\ref{condition3}), and the
corresponding paring given by $\left< P , Q  \right> = I
(res(PQ))$ for all $P,Q \in \Psi(R_z)$, is non-degenerate. 
We can then apply Proposition \ref{kpp} and obtain the following result:
\begin{Corollary} \label{cordef}
    Let us equip the Lie algebra $\Psi(R_z)$ with the non-degenerate pairing
    $\left< P , Q  \right> = I(res(PQ))$ for all $P,Q \in \Psi(R_z)$.
    Write $\Psi(R_z) = \Psi^{-1}(R_z) \oplus \mathcal{D}_{R_z}$
    and consider the Hamiltonian functions
    \begin{equation}
    \mathcal{H}_{k}(L_z) = \frac{1}{k} \, I \left( \mbox{res}(L_z^{k+1}) \right) \; ,
    \; \; \;  L_z \in \Psi(R_z)' = \Psi(R_z) \; ,
    \end{equation}
    on $\Psi(R_z)'$. The corresponding Hamiltonian equations of
    motion with respect to the $\{ \, , \, \}_{0}$ Poisson structure
    of $\Psi(R_z)'$ are
    \begin{equation} \label{flow}
    \frac{d L_z}{d t_{k}} = \left[ (L_z^{k})_{+} , L_z \right] \; .
    \end{equation}
\end{Corollary}

As in the classical theory, see \cite{D} or the more recent review
\cite{ER2013}, the Hamiltonian functions $\mathcal{H}_k$ furnish
an infinite family of $\R[[z]]$-conservation laws for the flow of
(\ref{flow}). 

\subsection{The KP equation and perturbed solutions}

In this final subsection we consider briefly the perturbed solutions $\tilde{L_z} \in
C^\infty(T, \Psi^1((R_z)_q))$ to (\ref{lolo4}) obtained by using the
second proof of Theorem \ref{KPcentral}, and we discuss how to connect them with solutions
to the ($q$-deformed)  KP-II equation. Now, in writing $\tilde{L_z}$ we are following
the notations of Subsection 4.2 and Corollary \ref{cordef}. In order to avoid complicated 
symbols, hereafter we write $L$ instead of $\tilde{L_z}$, and we also set $R_{z,q}= (R_z)_q.$ 
We remark once again (see Subsection 4.2) that derivatives with respect to $t_k$ are 
derivatives of smooth functions, and also that within this approach the deformation 
parameter $q$ is mandatory, because we have to use fully regular Lie groups
as in Lemma 4.4 and Theorem 4.5. We first observe:

\begin{Theorem} \label{kp2}
Let $L$ be a solution to $(\ref{lolo4})$. The 
pseudo-differential operator $L$ satisfies the system of partial differential equations
\begin{equation} \label{zcc3}
\frac{\d (L^i_+)}{\d t_j}- \frac{\d (L^j_+)}{\d t_i}+
\left[L^i_+,L^j_+\right] = 0 \; .
\end{equation}
\end{Theorem}
\begin{proof}
The $q=1$ case is, for instance, in Dickey's book \cite{D}. The general case is proven in
a similar way: 
we first observe that the equation
$$\frac{d L^m}{d t_n} =  [ L^n_+ , L^m ]$$
holds. Then, we use this equation and (\ref{lolo4}) to check that the expression
\begin{eqnarray*}
\left( \frac{d}{d t_n}(L^m_+) - \frac{d}{d t_m}(L^n_+) - [L^n_+ , L^m_+] 
\right)  -   \\
\left\{ \frac{d}{d t_n}(-L^m_-) - \frac{d}{d t_m}(-L^n_-) - [-L^n_- , -L^m_-] 
\right\} 
\end{eqnarray*}
equals zero. Since the round bracketed expression is a differential
operator, while the curly bracketed expression is an integral operator, both must be zero.  
\end{proof}

We consider $j=2$ and $i=3$, and we fix $t_1$ and $t_n$ for $n \geq 4.$ Our definition 
(\ref{sca}) yields 
\begin{eqnarray*}
L^2_+ & = & q^2 \left( \d^2 + 2 u_1(x,q t_1, q^2 t_2,q^3t_3, \cdots) \right)  \\
L^3_+ & = & q^3 \left( \d^3 + 3\, u_1(x,q t_1, q^2 t_2,q^3t_3, \cdots)\, \d +
            3\, u_{1,x}(x,q t_1, q^2 t_2,q^3t_3, \cdots) \right. \\
      &   & + \left. 3\, u_2(x,q t_1, q^2 t_2,q^3t_3, \cdots) \right) \, . 
\end{eqnarray*}
We replace these expressions into Equation (\ref{zcc3}) and we set $ t_2 = y, t_3 = t$, 
$2u_{-1}=u$, so that the functions $u$ and $u_2$ depend on $x$, $q^2 y$ and $q^3 t$, and on   
extra parameters $q t_1$, $q^4 t_4 , q^5 t_5, \cdots$. Note that the variable $x$ has not been scaled. In other words, we do not identify 
the space variable $x$ with $t_1$ as it is sometimes done, see
for instance \cite{D}. Equation (\ref{zcc3}) yields, after canceling a common $q^5$ factor,
\begin{eqnarray*}
u_{xx} + 4  u_{2,x} & = & u_y \\
\frac{1}{3} u_{xxx} + 2 u_{2,xx} - u_x u & = & u_{xy} + 2 u_{2,y} - \frac{2}{3}u_t\; .
\end{eqnarray*}
We eliminate $u_2$ from these equations and we obtain 
\begin{equation} \label{qkpii}
\frac{1}{2} u_{yy} = \d_x \left( -\frac{1}{6} u_{xxx} -  u u_x + \frac{2}{3} u_t \right)  \; .
\end{equation} 
This is the standard KP-II equation (for the function $u(x,q^2 y,q^3 t)\,$) as it appears 
in \cite{Bourg1993,D,M1}.

\smallskip

We would like to compare briefly our viewpoint to other approaches to KP-II. Let us consider, for definitiveness, the particular initial condition 
$$ L_0 = \d + u_0\d^{-1} $$ where $u_0 \in R_z\,$.  Proceeding as in subsection
4.2, we obtain the solution 
$$
L = \d + u_{-1} \d^{-1 } + (\hbox{lower order terms}) \in C^\infty(\K^\infty,\Psi_q(R_{z})) 
$$
to (\ref{lolo4}), and therefore Theorem \ref{kp2} and our foregoing discussion imply that $u = 2\,u_{-1} \in R_{z,q}$ solves 
the $q-$deformed KP-II equation (\ref{qkpii}). 
We stress that $u$ is a ``perturbed solution" depending on the extra parameters $z$ and $q$. The 
$z$-parameter is connected with the fact that we are solving KP with coefficients with low 
regularity, as we explain in subsection 5.2; the $q$-parameter is a ``time-scaling" 
parameter motivated by geometric considerations, as pointed out before.  Intuitively, in order
to interpret our derivations and series, we may assume that we are in the following situation:

\begin{itemize}
	\item[~] The points $(x,y) \in\;  ]-\pi;\pi]^2 \;\sim\; S^1 \times ]-\pi;\pi] \sim 
	\mathbb{T}^2$, the operators $\d_x,\d_y$ are classical derivatives with respect to 
	coordinates, and $\d_x^{-1}$ is a well defined primitive 
	map. By expansion on (e.g. Fourier) series of a map 
	$\phi \in C^\infty(\,]-\pi;\pi],\K),$ we obtain an element $u_0 \in R_z$ 
	which converges to $\phi$ when $z \rightarrow 1,$ uniformly on any compact 
	subset of $]-\pi;\pi].$ 
\end{itemize}

Let us consider the operator $S_0$ connecting the initial condition $L_0 = \d + u_0\d^{-1}$ to $\d = \d_x.$  We solve the equation $q L_0 = S_0 \cdot (q \d ) \cdot S_0^{-1}$ in which 
$$S_0 = 1 + \sum_{k \in \N^*} s_{-n}\d^{-n}\; ,$$
for the ``deformed initial data'' $q L_0\,$.  Expanding both sides of the equation $L_0 \, S_0 = S_0\, \d$ we obtain 
\begin{equation} \label{qic}
 s_{-1,x} = - u_0 \; ,
 \end{equation}
and proceeding recursively we find $s_{-n}$ in terms of $u_0$ and its derivatives.
%, as $z$-series of smooth functions defined on $]-\pi,\pi[\;$ which do not need to 
%converge to smooth functions on 	$S^1$. 
Equation (\ref{qic}) implies 
\begin{equation} \label{qicd}
\int_{S^1}u_0 = 0 \; .
\end{equation}
This necessary condition on the function $u_0$ appears also in 
\cite[Equation (1.8)]{Bourg1993}, as condition on the initial data for the KP-II equation. 

\smallskip

Summarizing, our solution $u=2u_{-1}$ to the $q$-deformed KP-II equation (\ref{qkpii}) is 
a $q$-series with coefficients in $C^\infty(\K^\infty,R_z)$ ---in which we recall that 
$R_z = C^\infty(S^1,\K)[[z]]$--- and initial data $u_0 \in R_z$ satisfying (\ref{qicd}). 
In other words, we begin with a ``rough function of $x$" as our given data (our $z$-series 
$u_0$), and we obtain a unique solution $u(x,q^2 y,q^3 t)$ for which {\em the full 
$y$-dependence, including at $t=0$, is obtained by integration of the hierarchy}. This 
certainly is different to what Bourgain proves in \cite{Bourg1993} for standard KP-II: 
in this reference he shows that the Cauchy problem for KP-II is globally well-posed for 
{\em arbitrary initial data} $u(x,y,0) \in L^2(S^1\times S^1)$. Thus, it would seem that
we can reach only 
a restricted class of solutions to KP-II\footnote{The following example was essentially 
suggested by the referee: let us fix a smooth function $u_0:S^1 \rightarrow \K$ and think 
of it as a (trivial) element of $R_z = C^\infty(S^1,\K)[[z]]$. Then, we obtain a unique solution $u = \sum \sum f_n^q(x,qt_1,q^2y,q^3t, qt_4, 
\cdots)z^n q^k$ to the $q$-deformed KP-II equation (\ref{qkpii}), with $y$ dependence of u
fixed by integration. On the other hand, let us choose a smooth function $f$ on $S^1$ 
with $f(0)=0$ and such that $\tilde{u}_0(x,q^2 y) = q u_0(x)+ q f(q^2 y)$ is different from 
$u(x,0,q^2 y,0,\cdots)$. Then, the solution to KP-II with initial data $(1/q)\tilde{u}_0$ 
arising via Bourgain's theorem \cite{Bourg1993} induces, as explained in subsection 4.2, a
solution $\tilde{u}$ to our $q$-deformed KP-II equation. The function $\tilde{u}$ cannot 
coincide with our solution $u$, even though they coincide at $t=y=0$.}; we wonder if these 
solutions can be compared with the ``meromorphic'' solutions 
appearing in \cite{Sh1986}, in which one (complex) variable is required for the solutions 
described therein. These matters are currently under investigation.

The foregoing observations are not restricted to KP-II. In fact, we could conclude from 
the interesting paper \cite{DGRW} and our previous remarks that the KP hierarchy
furnishes only some of the solutions to the equation
$$
w_{xxxy} + 3 w_{xy}w_x + 3 w_y w_{xx} + 2w_{yt} - 3 w_{xz} = 0
\; ,
$$ 
precisely the next equation after KP-II which can be deduced using Theorem \ref{kp2}.
This discussion leads us to hypothesize on the existence of a possible equation extending 
the KP hierarchy, which could furnish a full description of all the solutions of the KP-II 
equation. There is already a candidate for this hypothetical equation, as it is a known 
fact that the KP hierarchy may be seen as a  special case of a self-dual Yang-Mills 
hierarchy, see \cite{ACT,Sc}. This situation seems to be a non-trivial
instance of the well-known problem of differential inclusion, and
this again is reason for us to suspect the existence of another
equation, perhaps as the ones considered in  \cite{ACT,Sc},  
for which the KP hierarchy would be a specialization.

\vskip 12pt

\paragraph{\bf Acknowledgements:} Both authors have been partially
supported by CONICYT (Chile) via the Fondo Nacional de Desarrollo
Cient\'{\i}fico y Tecnol\'{o}gico operating grant \# 1161691.


\begin{thebibliography}{99}
    \bibitem{ACT}  Ablowitz, M; Chakravarty, S.; Takhtajan, L.A.; A Self-Dual Yang-Mills Hierarchy
                            and its Reductions to Integrable Systems
                            in 1 + 1 and 2 + 1 Dimensions; {\em Commun. Math. Phys.} 158 (1993), 289--314.
    \bibitem{ARS1} Adams, M.; Ratiu, T.; Schmidt, R.; A Lie group structure for pseudo-differential operators;
                  \textit{Math. Annalen} \textbf{273}, no 4 (1986) 529-551
    \bibitem{ARS2} Adams, M.; Ratiu, T.; Schmidt, R.; A Lie group structure for Fourier integral operators;
                   \textit{Math. Annalen} \textbf{276}, no.1 (1986), 19--41.
    \bibitem{Ad} Adler, M.; On a Trace Functional for Formal Pseudo-Differential Operators
                 and the Symplectic Structure of the Korteweg-Devries Type Equations. {\em 
                 lnventiones math.}  \textbf{50} (1979), 219--248.
\bibitem{BN2005} Batubenge, A.; Ntumba, P.; On the way to Fr\"olicher Lie groups
                     \textit{Quaestionnes mathematicae} \textbf{28} no1 (2005), 73--93
\bibitem{BR1984} Beals, M. and Reed, M.; Microlocal regularity theorems for nonsmooth 
    pseudodifferential operators and applications to nonlinear problems; \textit{Trans. Amer.
     Math. Soc.} \textbf{285} (1984), 159-184.
\bibitem{BR} Bernshtein, N.; Rozenfel'd, B.I.; Homogeneous spaces of infinite-dimensional Lie 
algebras and characteristic classes of foliations. {\em Russ. Math. Surv.} \textbf{28} (19783), 
107--142.
\bibitem{B} Bourbaki, N.; {\it Elements of Mathematics}. Springer-Verlag, Berlin. (1998)
    \bibitem{Bourg1993} Bourgain, J.; On the Cauchy problem for the Kadomtsev-Petviashvili equation. \textit{Geom. Funct. Anal.} \textbf{3} no 4 (1993) 315--341
    \bibitem{CW} Christensen, D. and Wu, E.; Tangent spaces and tangent bundles for
                 diffeological spaces, \emph{Cahiers de Topologie et G\'eom\'etrie Diff\'erentielle} Volume \textbf{LVII} (2016),
		3--50.
    \bibitem{D1} Demidov, E.E.; On the Kadomtsev-Petviashvili hierarchy with a noncommutative
                 timespace. {\em Funct. Anal. Appl.} \textbf{29} no. 2, (1995), 131--133.
    \bibitem{D2} Demidov, E.E.; Noncommutative deformation of the Kadomtsev-Petviashvili hierarchy.
                 In ``Algebra. 5, Vseross. Inst. Nauchn. i Tekhn. Inform. (VINITI)'', Moscow, 1995.
                (Russian). {\em J. Math. Sci. (New York)} \textbf{88} no. 4, (1998), 520--536 (English).
\bibitem{D} Dickey, L.A.; \textit{Soliton equations and Hamiltonian systems, second edition} 
(2003). Advanced Series in Mathematical Physics $12$, World Scientific Publ. Co., Singapore.
\bibitem{DGV} Dodson, C.; Galanis, G.; Vassiliou, E.; \textit{Geometry in the Fr\'echet context: 
a projective limit approach} London Mathematical Society Lecture Notes Series {\bf 428}, 
Cambridge University Press (2015)
    \bibitem{DGRW} Dorizzi, G.; Grammaticos, B.; Ramani, A.; Winternitz,
                   P.; Are all the equations of the
                   Kadomtsev-Petviashvili hierarchy integrable?
                   {\em J. Math. Phys.} {\bf 27} (1986),
                   2848--2852.
    \bibitem{DN2007-1} Dugmore, D.; Ntumba, P.;On tangent cones of Fr\"olicher spaces
                      \textit{Quaetiones mathematicae} \textbf{30} no.1 (2007),
                      67--83.
\bibitem{ER2013} Eslami Rad, A.; Reyes, E. G.; The Kadomtsev-Petviashvili hierarchy and the 
Mulase factorization of formal Lie groups {\it J. Geom. Mech.} {\bf 5}, no 3 (2013) 345--363.
\bibitem{FK} Fr\"olicher, A; Kriegl, A; {\it Linear spaces and differentiation theory} Wiley 
series in Pure and Applied Mathematics, Wiley Interscience (1988).
\bibitem{F} Fr\"olicher, A; Cartesian closed categories and analysis of smooth maps. In:
'Categories in Continuum Physics', F.W. Lawvere and S.H. Schanuel (Eds.) LNM 1174 (1986), 
Springer-Verlag, Berlin.
\bibitem{GV} Galanis, G. and Vassiliou, E.; A generalized frame bundle for certain
 Fr\'echet vector bundles and linear connections. \textit{Tokyo J. Math.} {\bf 20}  no.1 (1997),
                 129--137.
\bibitem{GD1981} Gelfand, I.M.; Dorfman, I.Y.; Infinite dimensional operators and infinite 
dimensional Lie algebras \textit{Funk. Anal. Priloz.} {\bf 15} (1981)
                     23--40.
\bibitem{Glo} Gl\"ockner, H; Algebras whose groups of the units are Lie groups
                  \textit{Studia Math. } \textbf{153}, no2 (2002), 147--177.
\bibitem{GR} Guieu, L. and Roger, C.; \textit{L' Alg\`ebre et le
           groupe de Virasoro: Aspects g\'eom\'etriques et
           alg\'ebraiques, generalisations}. Centre de Recherches Mathematiques, Universit\'e de
           Montreal, (2007).
\bibitem{Igdiff} Iglesias-Zemmour, P.; \textit{Diffeology} AMS mathematical monographs {\bf 185} 
(2013)
\bibitem{KM} Kriegl, A.; Michor, P.W.; \textit{The convenient setting for global analysis}
Math. surveys and monographs \textbf{53}, American Mathematical society, Providence, USA. (2000)
\bibitem{KO1990} Khesin, B.A. and Ovsienko, V.Y.; Symplectic leaves of the Gelfand-Dickii 
brackets and homotopy casses of non flattening curves. \textit{Funk. Anal. Prihoz.} {\bf 24}   
(1990), 38--47.
    \bibitem{KZ} Khesin, B.A. and Zakharevich, I.; Poisson-Lie groups of pseudodifferential
                 symbols, {\em Comm. Math. Phys.} \textbf{171} no. 3 (1995), 475--530.
\bibitem{KW} Khesin, B.A. and Wendt, R.; ``The geometry of infinite-dimensional groups" (2009).
 Springer-Verlag, Berlin.
\bibitem{Ku} Kubo, F.; Non-commutative Poisson algebra structures on affine Kac-Moody algebras.
                 {\em J. Pure and Applied Algebra} \textbf{126} (1998),
                 267--286.
\bibitem{Les} Leslie, J.; On a Diffeological Group Realization of certain Generalized 
symmetrizable Kac-Moody Lie Algebras \textit{J. Lie Theory} \textbf{13} (2003), 427--442.
\bibitem{Ma} Magnot, J-P.; Diff\'eologie du fibr\'e d'Holonomie en dimension infinie,
                 \textit{ C. R. Math. Soc. Roy. Can.} \textbf{28} (2006), 121--127.
\bibitem{Ma2013} Magnot, J-P.; Ambrose-Singer theorem on diffeological bundles and complete 
integrability of KP equations. {\em Int. J. Geon. Meth. Mod. Phys.} {\bf 10}, no 9 (2013) 
Article ID 1350043.
\bibitem{Ma2015} Magnot, J-P.; {\it q-deformed Lax equations and their differential geometric 
background}, Lambert  Academic Publishing, Saarbrucken, Germany (2015).
\bibitem{Ma2018-2} Magnot, J-P.; The group of diffeomorphisms of a non-compact manifold is 
not regular \textit{Demonstr. Math.} 51, No. 1, 8-16 (2018)              
\bibitem{Marsch1988} Marschall, J.; Pseudo-differential operators with coefficients in Sobolev 
spaces; {\em Trans. AMS } \textbf{307}, no1 (1988), 335-361.
    \bibitem{MR} Marsden, J. and Ratiu, T.; ``Introduction to mechanics and symmetry.
                 A basic exposition of classical mechanical systems. Second edition'' (1999).
                 Texts in Applied Mathematics, 17. Springer-Verlag, New York. 
\bibitem{M1} Mulase, M.; Complete integrability of the Kadomtsev-Petvishvili equation.
                {\em Advances in Math.} \textbf{54} (1984), 57--66.
\bibitem{M3} Mulase, M.; Solvability of the super KP equation and a generalization of the 
    Birkhoff decomposition. {\em Invent. Math.} \textbf{92} (1988), 1--46.
\bibitem{M4} Mulase, M.; Algebraic Theory of the KP Equations. In: ``Perspectives in 
Mathematical Physics" R. Penner and S.T. Yau (Eds.). International Press, Boston, 157--223 
(1994).
\bibitem{Neeb2007} Neeb, K-H.; Towards a Lie theory of locally convex groups 
    \textit{Japanese J. Math.} \textbf{1} (2006), 291-468
    \bibitem{O}  Olver, P.J.; ``Applications of Lie Groups to Differential
                 Equations'' (1993). Second Edition, Springer-Verlag, New York.
\bibitem{Om} Omori, H.; \textit{Infinite dimensional Lie groups} AMS translations of 
    mathematical monographs \textbf{158} (1997)
\bibitem{OMYK2} Omori, H.;  Maeda, Y.;  Yoshioka, A.; Kobayashi,  O.; { On regular Fr\'echet
               Lie groups IV }; {\it Tokyo J. Math.} {\bf 5}  (1981),
               365-397.
\bibitem{P} Perelomov, A.M.; ``Integrable systems of classical mechanics and Lie algebras'' 
(1990) Birkh\"{a}user Verlag, Berlin.
\bibitem{PS} Pressley, A. and Segal, G.B.; ``Loop groups" (1986). Oxford University Press.
\bibitem{RS} Reyman, A.G. and Semenov-Tian-Shansky, M.A.; Reduction of Hamiltonian Systems,
Affine Lie Algebras and Lax Equations II. {\em Invent. math.} \textbf{63} (1981), 423--432.
\bibitem{R-ams}  Reyes, E.G.; Jet Bundles, Symmetries, Darboux Transforms. In: 'Algebraic Aspects of Darboux Transformations, Quantum Integrable Systems and Supersymmetric Quantum Mechanics', P.B. Acosta-Humanez, F. Finkel, N. Kamran, P.J Olver (Eds.), Contemporary Mathematics 563, AMS, 137--164, (2012).  
\bibitem{Rob} Robart, T.; Sur l'int\'egrabilit\'e des sous-alg\`ebres de Lie en dimension 
infinie; \textit{Can. J. Math.} \textbf{49} (4) (1997), 820-839.
\bibitem{Sc} Schiff, J.; Self-Dual Yang-Mills and the Hamiltonian Structures of Integrable 
Systems. IAS preprint IASSNS-HEP-92/34.  \texttt{ArXiv: hep-th/9211070}.
\bibitem{STS} Semenov-Tian-Shansky, M.A.; What is a classical $r$-matrix? {\em Funct. Anal. 
Appl.} \textbf{17} (1983), 259--272.
\bibitem{Sh1986} Shiota, T.; Characterization of Jacobian varieties in terms of Soliton 
equations \textit{Inv. Math.} {\bf 83} (1986), 333--382.
\bibitem{Sou} Souriau, J.M.; Un algorithme g\'en\'erateur de structures quantiques;
                  \textit{Ast\'erisque}, Hors S\'erie, 341-399 (1985).
\bibitem{Stern} Sternberg, S.; ``Lectures on Differential Geometry" (1998). 2nd edition, AMS.
\bibitem{W} Waliszewski, W.; Jest in differentiable spaces. {\em \v{C}asopis pro p\v{e}stov\'{a}n\'{i} matematiky} \textbf{110} (1985), 241--249.
    \bibitem{W1} Watanabe, Y.; Hamiltonian structure of Sato's hierarchy of KP equations and a
                 coadjoint orbit of a certain formal Lie group. {\em Lett. Math. Phys.} 
                 \textbf{7} (1983), 99--106.
    \bibitem{W2} Watanabe, Y.; Hamiltonian structure of M. Sato's hierarchy of
                 Kadomtsev--Petviashvili equation. {\em Ann. Mat. Pura Appl.} (4) \textbf{136} (1984), 77--93.
    \bibitem{Wa} Watts, J.; \textit{Diffeologies, differentiable spaces and symplectic geometry}. University of Toronto,
                 PhD thesis. arXiv:1208.3634v1.
\end{thebibliography}
\end{document}